\documentclass[3p,final]{elsarticle}

\biboptions{sort&compress}
\makeatletter
\def\ps@pprintTitle{%
 \let\@oddhead\@empty
 \let\@evenhead\@empty
 \def\@oddfoot{}%
 \let\@evenfoot\@oddfoot}
\makeatother

\newcommand\dd{\mathrm{d}}
\newcommand\ee{\mathrm{e}}
\newcommand\ii{\mathrm{i}}
\newcommand\Id{\mathrm{Id}}
\newcommand\wt{\widetilde}
\newcommand\bs{\boldsymbol{s}}
\newcommand\Sf{s_{\mathrm{f}}}
\newcommand\Si{s_{\mathrm{i}}}
\newcommand\mQ{\mathcal{Q}}
\newcommand\Sarr{s_{\uparrow}}

\usepackage{physics, amssymb, amsthm}
\usepackage{xcolor}

\DeclareMathOperator\sgn{sgn}
\DeclareMathOperator\lspan{span}

\newtheorem{definition}{Definition}
\newtheorem{theorem}{Theorem}
\newtheorem{lemma}{Lemma}
{\theoremstyle{remark} \newtheorem{remark}{Remark}}

\newcommand{\mc}[1]{\mathcal{#1}}
\newcommand{\mf}[1]{\mathfrak{#1}}

\newcommand\added[1]{{#1}}

\title{Inchworm Monte Carlo method for open quantum systems}

\author[]{Zhenning Cai\fnref{fn1}}
\ead{matcz@nus.edu.sg}

\author[]{Jianfeng Lu\fnref{fn2}}
\ead{jianfeng@math.duke.edu}

\author[]{Siyao Yang\fnref{fn1}}
\ead{siyao_yang@nus.edu.sg}

\fntext[fn1]{Department of Mathematics, National University of
  Singapore, Level 4, Block S17, 10 Lower Kent Ridge Road, Singapore 119076.}
\fntext[fn2]{Department of Mathematics, Department of Physics,
  Department of Chemistry, Duke University, Box 90320, Durham NC 27708, USA.}

\begin{document}

\begin{abstract}
  We investigate in this work a recently proposed diagrammatic
  quantum Monte Carlo method --- the inchworm Monte Carlo method --- for
  open quantum systems. We establish its validity rigorously based on
  resummation of Dyson series.  Moreover, we introduce an
  integro-differential equation formulation for open quantum systems,
  which illuminates the mathematical structure of the inchworm
  algorithm.  This new formulation leads to an improvement of the
  inchworm algorithm by introducing classical deterministic
  time-integration schemes. The numerical method is validated by
  applications to the spin-boson model.
\end{abstract}

\begin{keyword}
  quantum Monte Carlo \sep open quantum system \sep diagrammatic
  methods \sep spin-boson model
\end{keyword}

\maketitle

\section{Introduction}

For realistic quantum systems, the system we are interested in is
often coupled with an uninteresting environment to a non-negligible
extent, which requires us to study the open quantum system, including
the effects such as quantum decoherence \cite{Suominen2014} and
quantum dissipation \cite{Esposito2010}. The application of open
quantum system ranges in a wide variety of quantum fields, including
quantum optical systems \cite{Breuer2007}, nonlinear statistical
mechanics \cite{Lindenberg1990}, quantum computation
\cite{Nielsen2010}, etc. \added{In reality, no quantum system in experiment is
  completely isolated from the environment. The development of quantum
  technology also relies on mitigating and sometimes utilizing the
  quantum dissipation \cite{Plenio2002}, which calls for better
  understanding of open quantum systems both theoretically and
  numerically. The open quantum system can also arise from approximate algorithms for many-body quantum systems in condensed matter physics, such as the rather successful dynamical mean field theory (DMFT) approach which reduces a full many-body system to an impurity problem  which is then treated numerically as an open quantum system \cite{Gull2011RMP}.  }

\added{The challenge of numerical simulation of open quantum system
  lies in the coupling of the system of interest to a usually large
  bath which contains a huge number of degree of freedom. This makes the direct simulation of the whole system together with bath impossible. Thus it is necessary to carry out dimension reduction, hopefully to integrate out the degree of freedoms associated with the bath. Such issues of dimension reduction also arise in many other areas of modeling and simulations. Indeed, while we focus in this work on open quantum system, many of the ideas have their classical counterparts, some of them have led to active developments in applied mathematics, see e.g., \cite{Levermore1996, Chorin2000, Darve2009, Hijon2010, Nicoud2011, ChorinLu2015, LeiBakerLi2016, Torrilhon2016, Chen2017} and books \cite{EvansMorris1990, Kalnay2003, Garnier2009, E2011, ChorinHald2013}. It is our hope that study of dimension reduction for open quantum systems can also shed light on further development for dimension reduction of classical systems.}

\added{In the context of open quantum system, many dimension reduction techniques have been developed over the years.} 
One conventional approach is to apply
the Nakajima-Zwanzig projection operator technique to obtain an
integro-differential master equation \cite{Nakajima1958,
  Zwanzig1960}, \added{where the memory effects resulting from integrating out the bath degree of freedom are encoded by some memory kernel}. \added{The classical analog is the well-known Mori-Zwanzig formalism \cite{Mori1965,Zwanzig1961}}. By taking the weak-coupling limit, \added{the memory effect can be neglected to obtain a Markovian
approximation, known as the Lindblad equation \cite{Davies1974, Davies1976}, which is the quantum analog of the Langevin dynamics}. The Markovian approximation however
breaks down for open systems with stronger coupling, \added{so that one has to keep track of the non-Markovian dynamics of the projected density matrix describing the system of interests \cite{Breuer2007}}. \added{For classical systems, such dynamics is often modeled and studied as generalized Langevin equations \cite{Henery1972, Zwanzig1973}, while for quantum systems, the non-Markovian dynamics is much more complicated, and many works have been devoted to this topic.} \added{Analogously to the generalized Langevin equations, generalized quantum master equation with memory kernel has been used to model non-Markovian open quantum dynamics \cite{MeierTannor1999, ShiGeva2003, Kelly2013, Montoya-Castillo2016}, though it is often complicated to come up with a good estimate of the memory kernel, especially when the memory is long range.}

\added{Besides the quantum master equation framework, the non-Markovian evolution of the system can be also directly modeled and simulated using} the
path-integral approaches such as the QuAPI (quasi-adiabatic propagator
path integral) methods \cite{Makri1995, Makri1996} and the HEOM
(hierarchical equations of motion) technique
\cite{Strumpfer2012}. These methods yield accurate numerical results (\added{often referred as ``numerically exact'' in the literature}),
while the computational cost is extremely huge, often unaffordable. 
To reduce
the computational cost, one common strategy is to replace the exact summation
or numerical integration in these methods by  Monte
Carlo methods. In this paper, we are going to study a specific type of
path-integral methods called the diagrammatic quantum Monte Carlo
method \cite{Werner2009} to solve the time-dependent open quantum
systems.  In particular, our study is largely motivated by the
inchworm Monte Carlo method recently proposed in \cite{Cohen2015,
  Chen2017a} to reduce the variance in quantum Monte Carlo by
diagrammatic resummation.

The basis of the diagrammatic quantum Monte Carlo method has been established
as early as 1960s \cite{Keldysh1965}. However, as other quantum Monte Carlo
methods, this type of methods also suffer from the notorious dynamical sign
problem, meaning that the number of Monte Carlo samples is required to grow at
least exponentially in time in order to keep the accuracy of the simulation. To
relieve the dynamical sign problem, Stockburger and Grabert introduced
stochastic unraveling of influence functionals in \cite{StockBurger2002}, and
Makri \cite{Makarov1994, Makri2017} proposed to assume a finite memory time of
the bath-correlation function and apply an iterative procedure to efficiently
implement the summation. The inchworm Monte Carlo method applies the idea of
diagrammatic resummation as in the bold diagrammatic Monte Carlo method
\cite{Prokof'ev2007} to the real-time evolution of the quantum systems. Similar
to the bold-line diagrammatic Monte Carlo method proposed in \cite{Gull2010,
Gull2011} (see also a more mathematical presentation \cite{LiLu:BDMC}), the
inchworm Monte Carlo
method tackles the dynamical sign problem by lumping a large number of diagrams
into a ``bold line'', which effectively reduces the number of total diagrams to
be summed to reach desired numerical accuracy. %Using the iterative strategy,
The inchworm method makes maximum use of the previous calculations, at the
expense of higher memory cost for storing all Green's functions. Despite its
success in the application of spin-boson model \cite{Chen2017b} and the
Anderson impurity model \cite{Dong2017, Ridley2018}, it requires 
better understanding to reveal the intrinsic mathematical structure of the inchworm method, and to further improve the method. In particular, it would be interesting to see how the bold lines are built
on the basis of shorter bold lines, and how the bold lines propagate when the
iterative procedure is precise. While the answers to these questions are not
detailed in the original derivation of the inchworm method \cite{Cohen2015,
Chen2017a}, in this paper, we will show the validity of the inchworm method
with mathematical rigor. The rigorous proof not only justifies the original algorithm, but also leads us to a new formulation of the open quantum
system as an integro-differential equation, based on which more accurate and efficient numerical approaches can be developed. 

The inchworm Monte Carlo method and the new integro-differential
equation formulation will be proved to be applicable for the Ohmic
spin-boson model, which is a simple open quantum system widely used as
benchmark problems \cite{Wang2000, Kernan2002, Duan2017}. Based on the
integro-differential equation, part of the Monte Carlo integration can
be replaced by classical time-integration methods to achieve higher
accuracy. The resulting new algorithm will be applied to the spin-boson model
to show the numerical efficiency. \added{Possible extensions will also be  discussed. More generally, the inchworm method represents a wider class of model reduction and variance reduction techniques for open quantum systems, which are worth investigating from the mathematical point of view, as they offer some new challenges as well as new toolsets that might be applicable to other areas of study.} 

The rest of this paper is organized as follows. In Section \ref{sec:Dyson}, we
introduce the basic formulation of the open quantum system and its Dyson series
expansion. Section \ref{sec:inchworm} gives a complete review of the inchworm
Monte Carlo method and proves its validity. The integro-differential equation
associate with the inchworm algorithm is derived in Section \ref{sec:id}. As an
application, we analyze the spin-boson model in Section \ref{sec:spin-boson}.
Our new numerical method is introduced in Section \ref{sec:num} and some
numerical examples are given in Section \ref{sec:examples}. A simple summary is
given in Section \ref{sec:summary} as the end of the paper.

\section{Dyson series expansion for open quantum systems} \label{sec:Dyson}

Before considering open quantum system, let us first recall the
time-dependent perturbation theory and the associated Dyson
series. Consider the von Neumann equation for quantum evolution (of a
closed system)
\begin{equation} \label{eq:vonNeumann}
\ii \frac{\dd \rho}{\dd t} = [H, \rho],
\end{equation}
where $\rho(t)$ is the density matrix at time $t$, and $H$ is the
Schr\"odinger picture Hamiltonian with the form
\begin{displaymath}
H = H_0 + W.
\end{displaymath}
Here $H_0$ is the unperturbed Hamiltonian and $W$ is viewed as a perturbation.
Following the convention, for any Hermitian operator $A$, we define $\langle A
\rangle = \tr(\rho(0) A)$. We are interested in the evolution of the
expectation for a given observable $O$, defined by
\begin{equation} \label{eq:O(t)}
\langle O(t) \rangle = \tr(O \rho(t))
  = \tr(O \ee^{-\ii t H} \rho(0) \ee^{\ii t H})
  = \langle \ee^{\ii t H} O \ee^{-\ii t H} \rangle.
\end{equation}
Using standard time dependent perturbation theory, the unitary group $\ee^{-\ii
t H}$ generated by $H$ can be represented using a Dyson series expansion \cite{Dyson1949}
\begin{equation} \label{eq:Dyson}
\ee^{-\ii t H} = \sum_{n=0}^{+\infty}
  \int_{t > t_n > \cdots > t_1 > 0} (-\ii)^n
    \ee^{-\ii(t-t_n)H_0} W \ee^{-\ii(t_n - t_{n-1})H_0} W \cdots
  W \ee^{-\ii(t_2 - t_1)H_0} W \ee^{-\ii t_1 H_0}
    \,\dd t_1 \cdots \,\dd t_n,
\end{equation}
where the integral should be interpreted as
\begin{equation}
\int_{t > t_n > \cdots > t_1 > 0}
  \,\dd t_1 \cdots \,\dd t_n =
\int_0^t \int_0^{t_n} \cdots \int_0^{t_2}
  \,\dd t_1 \cdots \,\dd t_{n-1} \,\dd t_n.
\end{equation}
Inserting the Dyson series \eqref{eq:Dyson} into \eqref{eq:O(t)}, one obtains
\begin{equation} \label{eq:O}
\begin{split}
\langle O(t) \rangle &= \sum_{n=0}^{+\infty} \sum_{n'=0}^{+\infty}
  \int_{t > t_n > \cdots > t_1 > 0}
    \int_{t > t'_{n'} > \cdots > t'_1 > 0} (-\ii)^n \ii^{n'} \\
  & \quad \langle \ee^{\ii t'_1 H_0} W \ee^{\ii (t'_2 - t'_1) H_0}  W \cdots
      W \ee^{\ii (t'_{n'} - t'_{n'-1}) H_0} W \ee^{\ii (t-t'_{n'}) H_0}
      O \times {}\\
  & \quad \phantom{\langle} \ee^{-\ii (t-t_n) H_0} W
    \ee^{-\ii (t_n - t_{n-1}) H_0} W
    \cdots W \ee^{-\ii (t_2 - t_1) H_0} W \ee^{-\ii t_1 H_0} \rangle
    \,\dd t_1' \cdots \,\dd t_n' \,\dd t_1 \cdots \,\dd t_n.
\end{split}
\end{equation}
Since the unperturbed Hamiltonian $H_0$ is usually easier to solve,
the above expansion provides the basis of a feasible approach to find
$\langle O(t) \rangle$ using Monte Carlo method.

For notational simplicity, the above integral is often denoted by the
Keldysh contour plotted in Figure \ref{fig:Keldysh}. The Keldysh
contour should be read following the arrows in the diagram, and
therefore has a forward (upper) branch and a backward (lower)
branch. The symbols are interpreted as follows:
\begin{itemize}
\item Each line segment connecting two adjacent time points labeled by
  $t_{\mathrm{s}}$ and $t_{\mathrm{f}}$ means a propagator $\ee^{-\ii
  (t_{\mathrm{f}} - t_{\mathrm{s}}) H_0}$. On the forward branch,
  $t_{\mathrm{f}} > t_{\mathrm{s}}$, while on the backward branch,
  $t_{\mathrm{f}} < t_{\mathrm{s}}$.
\item Each black dot introduces a perturbation operator $\pm\ii W$, where we
  take the minus sign on the forward branch, and the plus sign on the backward
  branch. At the same time, every black dot also represents an integral with
  respect to the label, whose range is from $0$ to the adjacent label to its
  right.
\item The cross sign at time $t$ means the observable in the Schr\"odinger
  picture.
\end{itemize}
Note that according to the above interpretation, two Keldysh contours differ
only when at least one of the values of $n$, $n'$ and $t$ is different, while
the positions of the labels on each branch do not matter. Thus, by taking the
expectation $\langle \cdot \rangle$ of this ``contour'', we obtain the summand
in \eqref{eq:O}. Therefore $\langle O(t) \rangle$ can be understood as the sum
of the expectations of all possible Keldysh contours.
\begin{figure}[!ht]
\centering
\includegraphics[width=.45\textwidth]{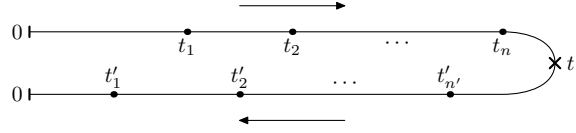}
\caption{Keldysh contour}
\label{fig:Keldysh}
\end{figure}

Such an interpretation also shows that we do not need to distinguish the
forward and backward branches when writing down the integrals. In fact, when
the series \eqref{eq:O} is absolutely convergent in the sense that
\begin{displaymath}
\begin{split}
%\langle O(t) \rangle =
 \sum_{n=0}^{+\infty} \sum_{n'=0}^{+\infty} &
  \int_{t > t_n > \cdots > t_1 > 0} \int_{t > t'_{n'} > \cdots > t'_1 > 0}
  \biggl\lvert \langle \ee^{\ii t'_1 H_0} W \ee^{\ii (t'_2 - t'_1) H_0}  W \cdots
      W \ee^{\ii (t'_{n'} - t'_{n'-1}) H_0} W \ee^{\ii (t-t'_{n'}) H_0}
      O \times {}  \\
  & \times  \ee^{-\ii (t-t_n) H_0} W \ee^{-\ii (t_n - t_{n-1}) H_0} W
    \cdots W \ee^{-\ii (t_2 - t_1) H_0} W \ee^{-\ii t_1 H_0} \rangle \biggr\rvert
    \,\dd t_1' \cdots \,\dd t_n' \,\dd t_1 \cdots \,\dd t_n < +\infty,
\end{split}
\end{displaymath}
we can reformulate \eqref{eq:O} as
\begin{equation} \label{eq:observable}
\begin{split}
\langle O(t) \rangle &= \sum_{m=0}^{+\infty}
  \int_{2t > s_m > \cdots > s_1 > 0} (-1)^{\#\{\bs < t\}} \ii^m \times {} \\
& \quad \times
  \left\langle G^{(0)}(2t, s_m) W G^{(0)}(s_m, s_{m-1}) W
  \cdots W G^{(0)}(s_2, s_1) W G^{(0)}(s_1, 0) \right\rangle
  \,\dd s_1 \cdots \,\dd s_m,
\end{split}
\end{equation}
where we use $\bs$ as a short-hand for the decreasing sequence
$(s_m, \cdots, s_1)$ and use $\#\{\bs < t\}$ to denote the number of elements in
$\bs$ which are less than $t$, \textit{i.e.}, the number of $s_k$ on
the forward branch of the Keldysh contour. For a given $t$, the
propagator $G^{(0)}$ is defined as
\begin{equation} \label{eq:G0}
G^{(0)}(\Sf, \Si) = \left\{ \begin{array}{ll}
  \ee^{-\ii (\Sf - \Si) H_0}, & \text{if } \Si \leqslant \Sf < t, \\
  \ee^{\ii (\Sf - \Si) H_0}, & \text{if } t \leqslant \Si \leqslant \Sf, \\
  \ee^{\ii (\Sf - t) H_0} O \ee^{-\ii (t - \Si) H_0},
    & \text{if } \Si < t \leqslant \Sf.
\end{array} \right.
\end{equation}
The integral \eqref{eq:observable} can also be understood diagramatically as the
``unfolded Keldysh contour'' plotted in Figure \ref{fig:UnfoldedKeldysh}. In
order to use only a single integral in \eqref{eq:observable}, we set the range
of the unfolded Keldysh contour to be $[0, 2t]$, and the mapping of time points
from the unfolded Keldysh contour to the original Keldysh contour has been
implied in the definition of $G^{(0)}(\cdot,\cdot)$. By comparing
\eqref{eq:observable} with Figure \ref{fig:UnfoldedKeldysh}, one can see that
$G^{(0)}(\cdot,\cdot)$ can be considered as the unperturbed propagator on the
unfolded Keldysh contour, with an action of observable $O$ at time $t$.
\begin{figure}[!ht]
\centering
\includegraphics[width=.9\textwidth]{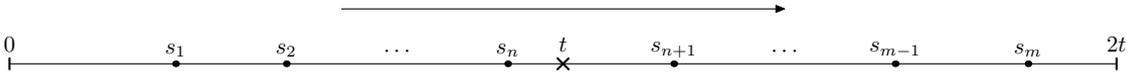}
\caption{Unfolded Keldysh contour}
\label{fig:UnfoldedKeldysh}
\end{figure}

To proceed, we now assume that the von Neumann equation
\eqref{eq:vonNeumann} describes an open quantum system coupled with a
bath, which means that both $\rho$ and $H$ are Hermitian operators on
the Hilbert space $\mc{H} = \mc{H}_s \otimes \mc{H}_b$, with
$\mc{H}_s$ and $\mc{H}_b$ representing respectively the Hilbert spaces
associated with the system and the bath. We let $H_0$ be the Hamiltonian
without coupling:
\begin{displaymath}
H_0 = H_s \otimes \mathrm{Id}_b + \mathrm{Id}_s \otimes H_b,
\end{displaymath}
where $H_s$ and $H_b$ are respectively the uncoupled Hamiltonians for
the system and the bath, and $\Id_s$ and $\Id_b$ are respectively the
identity operators for the system and the bath. Then the perturbation
$W$ describes the coupling, and here we assume $W$ takes the form
without loss of generality
\begin{displaymath}
W = W_s \otimes W_b.
\end{displaymath}
Furthermore, we assume the initial density matrix has the separable
form $\rho(0) = \rho_s \otimes \rho_b$, and we are concerned with observables
acting only on the system $O = O_s \otimes \Id_b$ (recall that physically the
system is the interesting part). With these assumptions, \eqref{eq:observable}
becomes
\begin{equation} \label{eq:observable1}
\langle O(t) \rangle = \sum_{m=0}^{+\infty}
  \ii^m \int_{2t > s_m > \cdots > s_1 > 0}
    (-1)^{\#\{\bs < t\}} \tr_s(\rho_s \mathcal{U}^{(0)}(2t, \bs, 0))
    \mathcal{L}_b(\bs) \,\dd s_1 \cdots \,\dd s_m,
\end{equation}
where the integrand is separated into $\mc{U}^{(0)}$ and $\mc{L}_b$ for the
system and bath parts:
\begin{align}
\label{eq:mc_U}
\begin{split}
\mc{U}^{(0)}(\Sf, \bs, \Si) &=
  \mc{U}^{(0)}(\Sf, s_m, \cdots, s_1, \Si) \\
&= G_s^{(0)}(\Sf, s_m) W_s G_s^{(0)}(s_m, s_{m-1}) W_s
  \cdots W_s G_s^{(0)}(s_2, s_1) W_s G_s^{(0)}(s_1, \Si),
\end{split} \\
\label{eq:mc_L}
\mathcal{L}_b(\bs) &=
  \tr_b(\rho_b G_b^{(0)}(2t, s_m) W_b G_b^{(0)}(s_m, s_{m-1}) W_b
  \cdots W_b G_b^{(0)}(s_2, s_1) W_b G_b^{(0)}(s_1, 0)),
\end{align}
where $\tr_s$ and $\tr_b$ take traces of the system and bath 
respectively. The propagators $G_s^{(0)}$ and $G_b^{(0)}$ are defined
similarly to \eqref{eq:G0}:
\begin{equation}
  G_s^{(0)}(\Sf, \Si) =
  \begin{cases}
    \ee^{-\ii (\Sf - \Si) H_s},
    & \text{if } \Si \leqslant \Sf < t, \\[5pt]
    \ee^{-\ii (\Si - \Sf) H_s},
    & \text{if } t \leqslant \Si \leqslant \Sf, \\[5pt]
    \ee^{-\ii (t - \Sf) H_s} O_s \ee^{-\ii (t - \Si) H_s},
    & \text{if } \Si < t \leqslant \Sf,
  \end{cases}
\end{equation}
and
\begin{equation}
  G_b^{(0)}(\Sf, \Si) = \begin{cases}
    \ee^{-\ii (\Sf - \Si) H_b},
    & \text{if } \Si \leqslant \Sf < t, \\[5pt]
    \ee^{-\ii (\Si - \Sf) H_b},
    & \text{if } t \leqslant \Si \leqslant \Sf, \\[5pt]
    \ee^{-\ii (2t - \Si - \Sf) H_b},
    & \text{if } \Si < t \leqslant \Sf.
    \end{cases}
\end{equation}
Note that the observable $O_s$ is inserted into the propagator
$G_s^{(0)}$ to keep the expression in \eqref{eq:mc_U} concise.

\section{Inchworm algorithm} \label{sec:inchworm}
In this section, we are going to study the inchworm algorithm introduced in
\cite{Chen2017a}, where the method was proposed for the spin-Boson model from a
purely diagrammatic point of view. By matching the mathematical interpretation
and the diagrammatic interpretation of the algorithm, we will establish
rigorously the validity of the algorithm in a more general sense. The central
idea of the algorithm is to consider the problem as an evolution problem and
reuse as much previous information as possible. We will start the
introduction of the algorithm by introducing the ``full propagators'', which
are exactly the carriers of the information to be recycled.

\subsection{Full propagator and its Dyson series expansion}
The inchworm algorithm proposed in \cite{Chen2017a} considers the following
``full propagators'':
\begin{equation} \label{eq:G}
G(\Sf, \Si) = 
\begin{cases}
  \tr_b(\rho_b G_b^{(0)}(2t, \Sf) \ee^{-\ii (\Sf - \Si) H} G_b^{(0)}(\Si, 0)),
    & \text{if } \Si \leqslant \Sf < t, \\[5pt]
  \tr_b(\rho_b G_b^{(0)}(2t, \Sf) \ee^{-\ii (\Si - \Sf) H} G_b^{(0)}(\Si, 0)),
    & \text{if } t \leqslant \Si \leqslant \Sf, \\[5pt]
  \tr_b(\rho_b G_b^{(0)}(2t, \Sf) \ee^{\ii (\Sf - t) H} O
    \ee^{-\ii (t - \Si) H} G_b^{(0)}(\Si, 0)),
    & \text{if } \Si < t \leqslant \Sf.
\end{cases}
\end{equation}
Here the trace is taken only on the space of the bath $\mc{H}_b$, and hence
$G(\Sf,\Si)$ is an operator on $\mc{H}_s$. Following the same method from
\eqref{eq:Dyson} to \eqref{eq:G0}, we get the following Dyson series expansion
for $G(\Sf, \Si)$:
\begin{equation} \label{eq:GDyson}
\begin{split}
G(\Sf, \Si) = \sum_{m=0}^{+\infty} \int_{\Sf > s_m > \cdots > s_1 > \Si}
  (-1)^{\#\{\bs < t\}} \ii^m & \tr_b \Big( \rho_b
  G_b^{(0)}(2t, \Sf) G^{(0)}(\Sf, s_m) W G^{(0)}(s_m, s_{m-1}) W \\
& \cdots W G^{(0)}(s_2, s_1) W G^{(0)}(s_1, \Si) G_b^{(0)}(\Si, 0)
  \Big) \,\dd s_1 \cdots \,\dd s_m.
\end{split}
\end{equation}
To apply the inchworm algorithm, we need the following two hypotheses, which are abstracted from the spin-Boson model studied in \cite{Chen2017a}:
\begin{enumerate}
\item[(H1)] \label{item:H1} The initial density matrix for the bath
  $\rho_b$ commutes with the Hamiltonian $H_b$. Physically, this condition holds
  when the bath is initially at the thermal equilibrium associated
  with the Hamiltonian $H_b$.
\item[(H2)] \label{item:H2} There exists a function $B(\cdot, \cdot)$ such that
  the following Wick's theorem holds:
  \begin{equation} \label{eq:Wick}
  \mc{L}_b(s_m, \cdots, s_1) = \left\{ \begin{array}{ll}
    0, & \text{if } m \text{ is odd}, \\[5pt]
    \displaystyle \sum_{\mf{q} \in \mQ(s_m, \cdots, s_1)} \mc{L}(\mf{q}),
      & \text{if } m \text{ is even},
  \end{array} \right.
  \end{equation}
  where the right hand side is given by all possible ordered pairings of the time points:
  \begin{align}
  \label{eq:Lq}
  & \mc{L}(\mf{q}) = \prod_{(\tau_1,\tau_2) \in \mf{q}} B(\tau_1, \tau_2), \\
  \begin{split}
  & \mQ(s_m, \cdots, s_1) =
    \Big\{ \{(s_{j_1}, s_{k_1}), \cdots, (s_{j_{m/2}}, s_{k_{m/2}})\} \,\Big\vert\,
    \{j_1, \cdots, j_{m/2}, k_1, \cdots, k_{m/2}\} = \{1,\cdots,m\}, \\
  & \hspace{220pt} s_{j_l} \leqslant s_{k_l} \text{ for any } l = 1,\cdots,m/2
  \Big\},
  \end{split}
  \end{align}
  When $m = 0$, the value of $\mc{L}(\emptyset)$ is defined as $1$.
\end{enumerate}
In hypothesis (H2), $\mQ(s_m, \cdots, s_1)$ is the set of all possible ordered
pairings of $\{s_m, \cdots, s_1\}$. For example,
\begin{align*}
& \mQ(s_2, s_1) = \bigl\{ \{(s_1,s_2)\} \bigr\}, \\
& \mQ(s_4, s_3, s_2, s_1) = \bigl\{ \{(s_1,s_2), (s_3,s_4)\},
  \{(s_1,s_3),(s_2,s_4)\}, \{(s_1,s_4),(s_2,s_3)\} \bigr\}.
\end{align*}
We can also represent these sets by diagrams:
\newcommand\insdiag[1]{\raisebox{-4pt}{\includegraphics[scale=.6]{#1}}}
\begin{align*}
& \mQ(s_2, s_1) = \left\{ \insdiag{images/PairSet.1} \right\}, \\
& \mQ(s_4, s_3, s_2, s_1) = \left\{ \insdiag{images/PairSet.2}, \quad
  \insdiag{images/PairSet.3}, \quad \insdiag{images/PairSet.4}
\right\}.
\end{align*}
Manifestly, each arc stands for a pair formed by the labels on the two end
points, and each diagram denotes a set of pairs.

As in the quantum field theory, Wick's theorem (hypothesis (H2)) turns
integrals into diagrams, which allows us to use diagrammatic quantum Monte
Carlo methods in the simulation. The first hypothesis (H1)
associates the full propagators with observables. Precisely, \added{when $\Si <
t < \Sf$ and $\Si + \Sf = 2t$, we have
\begin{displaymath}
G_b^{(0)}(\Si,0) \rho_b G_b^{(0)}(2t,\Sf) =
  \ee^{\ii (2t-\Sf) H_b} \rho_b \ee^{-\ii \Si H_b} =
  \rho_b \ee^{\ii (2t-\Sf-\Si) H_b} = \rho_b.
\end{displaymath}
Thus by the cyclic property of the operator trace, we can derive from
\eqref{eq:G} that}%
\begin{equation} \label{eq:obs}
\tr_s(\rho_s G(\Sf, \Si)) =
  \tr(\added{\rho(0)} \ee^{\ii (\Sf - t) H} O \ee^{-\ii (t - \Si) H}) =
  \langle O(t - \Si) \rangle,
\end{equation}
which shows that the evolution of the observable from time $0$ to $t$
can be fully obtained once the propagator $G(\Sf, \Si)$ is solved for every
pair of $\Sf$ and $\Si$. \added{In fact, when working with Wick's theorem (H2),
we usually assume that $\rho_b = \exp(-\beta H_b)$, and thus the hypothesis
(H1) is naturally fulfilled.} By splitting system and bath parts and
applying Wick's theorem \eqref{eq:Wick}, we get from \eqref{eq:GDyson} that
\begin{equation} \label{eq:DysonG}
G(\Sf, \Si) = \sum_{\substack{m=0\\[2pt] m \text{ is even}}}^{+\infty}
   \int_{\Sf > s_m > \cdots > s_1 > \Si}
   \sum_{\mf{q} \in \mQ(\bs)} (-1)^{\#\{\bs < t\}} \ii^m \mc{U}^{(0)}(\Sf, \bs, \Si)
   \mc{L}(\mf{q}) \,\dd s_1 \cdots \,\dd s_m.
\end{equation}
Here the integral of $(-1)^{\#\{\bs < t\}} \ii^m \mc{U}^{(0)}(\Sf, \bs, \Si) \mc{L}(\mf{q})$
can also be represented by a diagram like Figure \ref{fig:Dyson}, which is
interpreted by
\begin{itemize}
\item Each line segment connecting two adjacent time points labeled by
  $t_{\mathrm{s}}$ and $t_{\mathrm{f}}$ means a propagator
  $G_s^{(0)}(t_{\mathrm{f}}, t_{\mathrm{i}})$.
\item Each black dot introduces a perturbation operator $\pm\ii W_s$, and we
  take the minus sign on the forward branch, and the plus sign on the backward
  branch. Here the label for time $t$, which separates the two branches of the
  Keldysh contour, is omitted. Additionally, each black dot also represents the
  integral with respect to the label over the interval from $\Si$ to its next
  label.
\item The arc connecting two time points $t_{\mathrm{s}}$ and $t_{\mathrm{f}}$
  stands for $B(t_{\mathrm{s}}, t_{\mathrm{f}})$.
\end{itemize}
\begin{figure}[!ht]
\centering
\includegraphics[width=.45\textwidth]{images/Dyson.1}
\caption{Diagrammatic representation for the integral of $(-1)^{\#\{\bs < t\}}
\ii^m \mc{U}^{(0)}(\Sf, \bs, \Si) \mc{L}(\mf{q})$ when $m = 6$ and $\mf{q} =
\{(s_1,s_6), (s_2,s_4), (s_3,s_5)\}$.}
\label{fig:Dyson}
\end{figure}
\noindent Note that the branches are not explicitly labeled in Figure
\ref{fig:Dyson}. The two end points $\Si$ and $\Sf$ may both locate on the
forward branch or the backward branch; they may also belong to different
branches. Such a diagrammatic representation allows us to rewrite
\eqref{eq:DysonG} as
\begin{equation} \label{eq:DysonG_diagram}
\includegraphics[width=\textwidth]{images/Dyson.2}
\end{equation}
where the bold line on the left-hand side represents the full propagator
$G(\Sf,\Si)$. The Monte Carlo method based on \eqref{eq:DysonG_diagram} is referred to as ``bare diagrammatic quantum Monte Carlo'' method in \cite{Chen2017a}, which is essentially identical to the Monte Carlo method based on the Dyson series expansion \eqref{eq:Dyson}, except that the bath function $\mc{L}_b(\bs)$ is also evaluated by the Monte Carlo method.

\added{%
\begin{remark}
The hypothesis (H1) aims to establish the relation between the full propagators
and the observables. It can possibly be generalized for quantum Monte Carlo
algorithms in other circumstances. For example, in the Anderson impurity model
considered in \cite{Antipov2017}, the Hamiltonian $H(t)$ depends on time and
the initial state is the entire thermal equilibrium $\rho(0) = \exp(-\beta
H(0))$, where the Keldysh contour needs to be extended by adding the imaginary
Matsubara branch to get \eqref{eq:obs}. Another example is \cite{Ridley2018},
where the initial state includes chemical potentials, which does not exactly
satisfy (H1) but a relation similar to \eqref{eq:obs} still holds.
\end{remark}

\begin{remark} \label{rem:boson}
The hypothesis (H2) is a bosonic Wick's theorem in physical terms. For fermions, a power
of $-1$ needs to be added into the definition of $\mc{L}(\mf{q})$
\eqref{eq:Lq}. While  for definitness of presentation, we will stick to the bosonic representation in this work, the methodology
and analysis in this paper are generalizable to the fermionic case as well. In practice, the algorithms based
on bold diagrams have been successfully applied to the fermionic case such as
the Anderson impurity model \cite{Gull2011, Antipov2017}.
\end{remark}
}

\subsection{Description of the inchworm algorithm}
The inchworm algorithm uses another series expansion of $G(\Sf, \Si)$
that leads to efficient use of results of previous time steps for
future computations on the contour. Before introducing the inchworm series, we need
the following definitions:

\begin{definition}[Linked pairs and linked set of pairs]
Two pairs of real numbers $(s_1, s_2)$ and $(\tau_1, \tau_2)$ satisfying $s_1
\leqslant s_2$ and $\tau_1 \leqslant \tau_2$ are \emph{linked} if either of the
following two statements holds:
\begin{enumerate}
\item $s_1 \leqslant \tau_1 \leqslant s_2$ and $\tau_1 \leqslant s_2 \leqslant \tau_2$.
\item $\tau_1 \leqslant s_1 \leqslant \tau_2$ and $s_1 \leqslant \tau_2 \leqslant s_2$.
\end{enumerate}

For two sets of pairs $\mf{q}_1$ and $\mf{q}_2$, we say the two sets $\mf{q}_1$
and $\mf{q}_2$ are \emph{linked} if there exists $(s_1, s_2) \in \mf{q}_1$ and
$(\tau_1, \tau_2) \in \mf{q}_2$ such that $(s_1, s_2)$ and $(\tau_1, \tau_2)$
are linked. 

Given a set of pairs $\mf{q}$, we say $\mf{q}$ is a \emph{linked set of pairs}
if it cannot be decomposed into the union of two sets of pairs that are not
linked. We define
\begin{displaymath}
\mc{Q}_c(s_m,\cdots,s_1) = \{\mf{q} \in \mc{Q}(s_m,\cdots,s_1) \mid \mf{q} \text{ is linked}\}.
\end{displaymath}
\end{definition}

For example, according to the above definition, when $s_1 < s_2 < s_3 < s_4$,
the pairs $(s_1, s_3)$ and $(s_2, s_4)$ are linked, while $(s_1, s_4)$ and
$(s_2, s_3)$ are not, which can be also clearly seen from the diagrams:
\begin{displaymath}
\insdiag{images/PairSet.3} \quad \insdiag{images/PairSet.4}
\end{displaymath}
It is also clear that the set of pairs in the left diagram is a linked set,
while the other one is not.

\begin{definition}[Linked component decomposition]
  For any $\mf{q} \in \mQ(s_m, \cdots, s_1)$, there exists a collection of its
  disjoint subsets $\mf{q}_1, \mf{q}_2, \cdots, \mf{q}_n$ such that
\begin{enumerate}
\item $\mf{q}_1, \mf{q}_2, \cdots, \mf{q}_n$ are all linked sets of pairs;
\item $\mf{q} = \mf{q}_1 \cup \mf{q}_2 \cup \cdots \cup \mf{q}_n$;
\item $\mf{q}_{n_1}$ and $\mf{q}_{n_2}$ are not linked if $n_1 \neq n_2$.
\end{enumerate}
We call $\mf{q} = \mf{q}_1 \cup \mf{q}_2 \cup \cdots \cup \mf{q}_n$ the
\emph{linked component decomposition} of $\mf{q}$.
\end{definition}
An simple example of the linked component decomposition is given below by
diagrammatic representations:
\begin{displaymath}
\insdiag{images/PairSet.8}
\end{displaymath}

\begin{definition}[Inchworm properness] \label{def:inchworm}
Given a decreasing sequence of real numbers $s_m, \cdots, s_1$ and a real
number $\Sarr$, a set of integer pairs $\mf{q} \in \mQ(s_m, \cdots, s_1)$ is
called inchworm proper, if in its linked component decomposition $\mf{q} =
\mf{q}_1 \cup \mf{q}_2 \cup \cdots \cup \mf{q}_n$, each $\mf{q}_k$ contains at least
one point greater than or equal to $\Sarr$. Below we use $\mQ_{\Sarr}(s_m,
\cdots, s_1)$ to denote the collection of all inchworm proper pair sets.
\end{definition}

The structure of the set $\mQ_{\Sarr}(s_m, \cdots, s_1)$ is determined by the
relative position of $\Sarr$ in the decreasing sequence $s_m, \cdots, s_1$.
For example, if $m = 4$ and $\Sarr \in (s_2, s_3]$, then
\begin{displaymath}
\mQ_{\Sarr}(s_4, s_3, s_2, s_1) = \bigl\{
  \{(s_1,s_3),(s_2,s_4)\}, \{(s_1,s_4),(s_2,s_3)\} \bigr\} =
\left\{ \insdiag{images/PairSet.5}, \quad
  \insdiag{images/PairSet.6} \right\};
\end{displaymath}
if $\Sarr \in (s_3, s_4]$, then $\mQ_{\Sarr}(s_4, s_3, s_2, s_1)$ contains only
one element:
\begin{displaymath}
\mQ_{\Sarr}(s_4, s_3, s_2, s_1) = \bigl\{ \{(s_1,s_3),(s_2,s_4)\} \bigr\} =
\left\{ \insdiag{images/PairSet.7} \right\}.
\end{displaymath}
In fact, if $\Sarr \in (s_{m-1}, s_m]$, the number of sets of pairs in the linked component decomposition of an inchworm proper set $\mf{q}$ must be $1$, which means $\mf{q}$ is linked. Therefore we have
\begin{equation} \label{eq:linked_mQ}
\mQ_{\Sarr}(s_m, \cdots, s_1) = \mQ_c(s_m, \cdots, s_1),
  \qquad \text{if } s_{m-1} < \Sarr \leqslant s_m.
\end{equation}

Based on the definition of inchworm proper set of pairs, we have the following
theorem:
\begin{theorem} \label{thm:inchworm}
Suppose the Dyson series \eqref{eq:DysonG} is absolutely convergent in the
sense that
\begin{equation} \label{eq:abs_convergence}
\sum_{\substack{m=0\\[2pt] m \text{ is even}}}^{+\infty}
   \int_{\Sf > s_m > \cdots > s_1 > \Si} \|\mc{U}(\Sf, \bs, \Si)\|_s
     \left| \sum_{\mf{q} \in \mQ(\bs)} \mc{L}(\mf{q}) \right|
   \,\dd s_1 \cdots \,\dd s_m < +\infty, \quad \forall \Si, \Sf \in [0,2t],
\end{equation}
where $\|\cdot\|_s$ is the operator norm in the Hilbert space $\mc{H}_s$. For
any $\Sarr \in (\Si, \Sf)$, we have
\begin{equation} \label{eq:inchworm}
\begin{split}
G(\Sf, \Si) &= \mc{G}_{\Sarr}(\Sf,\Si) +
\sum_{\substack{m=2\\[2pt] m \text{ is even}}}^{+\infty}
  \int_{\Sf > s_m > \cdots > s_1 > \Si} (-1)^{\#\{\bs < t\}} \ii^m \left(
    \sum_{\mf{q} \in \mQ_{\Sarr}(\bs)} \mathcal{L}(\mf{q})
  \right) \times {} \\
& \mc{G}_{\Sarr}(\Sf, s_m) W_s \mc{G}_{\Sarr}(s_m, s_{m-1}) W_s
  \cdots W_s \mc{G}_{\Sarr}(s_2, s_1) W_s \mc{G}_{\Sarr}(s_1, \Si)
  \,\dd s_1 \cdots \,\dd s_m,
\end{split}
\end{equation}
where
\begin{displaymath}
\mc{G}_{\Sarr}(\Sf, \Si) = \left\{ \begin{array}{ll}
  G(\Sf,\Si), & \text{if } \Si \leqslant \Sf \leqslant \Sarr, \\
  G_s^{(0)}(\Sf,\Si), & \text{if } \Sarr < \Si \leqslant \Sf, \\
  G_s^{(0)}(\Sf,\Sarr) \, G(\Sarr,\Si),
    & \text{if } \Si \leqslant \Sarr < \Sf.
\end{array} \right.
\end{displaymath}
\end{theorem}

Note that the right-hand side of \eqref{eq:inchworm} has a very similar
structure to \eqref{eq:DysonG}, except that the bare Green's function
$G_s^{(0)}$ is replaced by $\mc{G}_{\Sarr}$ and the sum consists only of
``inchworm proper diagrams'', which will be further discussed below.  This
theorem shows that $G(\Sf,\Si)$ can be evaluated by the Monte Carlo simulation
of the right-hand side of \eqref{eq:inchworm} based on the knowledge of
$G(\tau_2, \tau_1)$ for all $\Si < \tau_1 < \tau_2 < \Sarr$.  Therefore the
algorithm can be designed as follows:
\begin{quote}
\hrule
Choose a time step $\Delta t = t/N$ \\
for $\Sf$ from $0$ to $2t$ with step $\Delta t$ \\
\hspace*{15pt} for $\Si$ from $\Sf$ to $0$ with step $-\Delta t$ \\
\hspace*{30pt} Use \eqref{eq:inchworm} to evaluate $G(\Sf,\Si)$ by the Monte
Carlo method \\
\hspace*{15pt} end for \\
end for \\[-5pt]
\hrule
\end{quote}
This algorithm computes all $G(\Sf, \Si)$ when $\Si$ and $\Sf$ are multiples of
$\Delta t$. When evaluating the right-hand side of \eqref{eq:inchworm},
interpolation might be needed to get $G(\tau_2, \tau_1)$ when $\tau_1$ or
$\tau_2$ is not an integer multiple of $\Delta t$. Details about the
interpolation will be given in Section \ref{sec:num}.  

A rigorous proof of Theorem \ref{thm:inchworm} will be given in the
next section. Here we would like to provide the diagrammatic
understanding of this algorithm, following \cite{Chen2017a,
  Chen2017b}.  As already mentioned, on the right-hand side of
\eqref{eq:inchworm}, only ``inchworm proper diagrams'' appear, for
example,
\renewcommand\insdiag[1]{\raisebox{-12pt}{\includegraphics[scale=.8]{#1}}}
\begin{equation} \label{eq:inchworm_diagram}
\begin{gathered}
\insdiag{images/Inchworm.1}, \qquad
\insdiag{images/Inchworm.2}, \qquad
\insdiag{images/Inchworm.3}, \\
\insdiag{images/Inchworm.4}, \qquad
\insdiag{images/Inchworm.5}, \qquad \cdots
\end{gathered}
\end{equation}
These diagrams are interpreted as follows:
\begin{itemize}
\item Each bold line segment connecting two adjacent time points labeled by
  $t_{\mathrm{i}}$ and $t_{\mathrm{f}}$ means a full propagator
  $G(t_{\mathrm{f}}, t_{\mathrm{i}})$.
\item Each thin line segment connecting two adjacent time points labeled by
  $t_{\mathrm{i}}$ and $t_{\mathrm{f}}$ means a propagator
  $G_s^{(0)}(t_{\mathrm{f}}, t_{\mathrm{i}})$.
\item Each black dot introduces a perturbation operator $\pm\ii W_s$, which
  takes the minus sign on the forward branch and the plus sign on the backward
  branch. It also introduces an integral with respect to the label from $\Sarr$
  to the next label.
\item Each white vertical line introduces a perturbation operator $\pm\ii W_s$,
  which takes the minus sign on the forward branch and the plus sign on the
  backward branch. It also introduces an integral with respect to the label
  from $\Si$ to the next label.
\item The arc connecting two time points $t_{\mathrm{i}}$ and $t_{\mathrm{f}}$
  stands for $B(t_{\mathrm{i}}, t_{\mathrm{f}})$.
\end{itemize}
By \eqref{eq:DysonG_diagram}, we can see that each bold line in the above
diagrams is a sum consisting of infinite ``thin diagrams'', which means that
\eqref{eq:DysonG_diagram} is a partial resummation of the Dyson series
expansion. Therefore it can be expected that \eqref{eq:inchworm} converges
faster than \eqref{eq:DysonG}. It is also worth mentioning that the following
diagrams are not included in the inchworm series
\eqref{eq:inchworm}:
\begin{displaymath}
\insdiag{images/Inchworm.6}, \qquad
\insdiag{images/Inchworm.7}, \qquad
\insdiag{images/Inchworm.8}, \qquad \cdots
\end{displaymath}
since the terms presented by these diagrams have actually appeared in other
diagrams, which are respectively
\begin{displaymath}
\insdiag{images/Inchworm.1}, \qquad
\insdiag{images/Inchworm.9}, \qquad
\insdiag{images/Inchworm.10}, \qquad \cdots
\end{displaymath}
Below we use a diagrammatic equation to summarize the idea of the inchworm algorithm:
\begin{equation} \label{eq:theorem}
\includegraphics[width=\textwidth]{images/Theorem.1}
\end{equation}
The proof of Theorem \ref{thm:inchworm} follows such understanding of
the inchworm series. However, to make sure no diagrams are missed or double-counted in this equation, we need to interpret these diagrammatic equations precisely using mathematical equations, which will be detailed in the following section.

\subsection{Proof of Theorem \ref{thm:inchworm}}
This section is devoted to the proof of \eqref{eq:inchworm}, which is to equate
the right-hand side of \eqref{eq:inchworm} and the series \eqref{eq:DysonG}. To
see this, we will first introduce $\Sarr$ to the expansion of $G(\Sf,\Si)$ by
the following lemma:

\begin{lemma}
When the Dyson series \eqref{eq:DysonG} is absolutely convergent in the sense
of \eqref{eq:abs_convergence}, for any $\Sarr \in (\Si, \Sf)$, it holds that
\begin{equation} \label{eq:DysonArrow}
\begin{split}
G(\Sf, \Si) = G_s^{(0)}(\Sf, \Sarr) G(\Sarr, \Si) +
  \sum_{\substack{m=2\\[2pt] m \text{ is even}}}^{+\infty}
  \sum_{p=0}^{m-1} & \int_{\Sf > s_m > \cdots > s_{p+1} > \Sarr}
  \int_{\Sarr > s_p > \cdots > s_1 > \Si} \sum_{\mf{q} \in \mQ(\bs)} \\
& (-1)^{\#\{\bs < t\}} \ii^m \mc{U}^{(0)}(\Sf, \bs, \Si)
  \mathcal{L}(\mf{q}) \,\dd s_1 \cdots \,\dd s_p \,\dd s_{p+1} \cdots \,\dd s_m.
\end{split}
\end{equation}
\end{lemma}

\begin{proof}
For any positive integer $m$, it holds that
\begin{equation} \label{eq:int_fact}
\int_{\Sf > s_m > \cdots > s_1 > \Si} \varphi(\bs) \,\dd s_1 \cdots \,\dd s_m
  = \sum_{p=0}^m \int_{\Sf > s_m > \cdots > s_{p+1} > \Sarr}
    \int_{\Sarr > s_p > \cdots > s_1 > \Si} \varphi(\bs)
      \,\dd s_1 \cdots \,\dd s_p \,\dd s_{p+1} \cdots \,\dd s_m
\end{equation}
for any function $\varphi$. This can be proven by mathematical induction since
when the above equation holds for some $m$, we have
\begin{displaymath}
\begin{split}
& \int_{\Sf > s_{m+1} > \cdots > s_1 > \Si}
  \varphi(\bs) \,\dd s_1 \cdots \,\dd s_{m+1} =
\int_{\Si}^{\Sf} \left( \int_{s_{m+1} > s_m > \cdots > s_1 > \Si}
  \varphi(\bs) \,\dd s_1 \cdots \,\dd s_m \right) \dd s_{m+1} \\
={} & \int_{\Si}^{\Sarr} \left( \int_{s_{m+1} > s_m > \cdots > s_1 > \Si}
  \varphi(\bs) \,\dd s_1 \cdots \,\dd s_m \right) \dd s_{m+1} +
\int_{\Sarr}^{\Sf} \left( \int_{s_{m+1} > s_m > \cdots > s_1 > \Si}
  \varphi(\bs) \,\dd s_1 \cdots \,\dd s_m \right) \dd s_{m+1} \\
={} & \int_{\Sarr > s_{m+1} > \cdots > s_1 > \Si}
  \varphi(\bs) \,\dd s_1 \cdots \,\dd s_{m+1} \\
& + \int_{\Sarr}^{\Sf} \left( \sum_{p=0}^m
  \int_{s_{m+1} > s_m > \cdots > s_{p+1} > \Sarr}
  \int_{\Sarr > s_p > \cdots > s_1 > \Si}
  \varphi(\bs) \,\dd s_1 \cdots \,\dd s_p \,\dd s_{p+1} \cdots \,\dd s_m \right)
  \dd s_{m+1} \\
={} & \int_{\Sarr > s_{m+1} > \cdots > s_1 > \Si}
  \varphi(\bs) \,\dd s_1 \cdots \,\dd s_{m+1} \\
& + \sum_{p=0}^m \int_{\Sf > s_{m+1} > \cdots > s_{p+1} > \Sarr}
  \int_{\Sarr > s_p > \cdots > s_1 > \Si}
  \varphi(\bs) \,\dd s_1 \cdots \,\dd s_p \,\dd s_{p+1} \cdots \,\dd s_{m+1} \\
={} & \sum_{p=0}^{m+1} \int_{\Sf > s_{m+1} > \cdots > s_{p+1} > \Sarr}
  \int_{\Sarr > s_p > \cdots > s_1 > \Si}
  \varphi(\bs) \,\dd s_1 \cdots \,\dd s_p \,\dd s_{p+1} \cdots \,\dd s_{m+1},
\end{split}
\end{displaymath}
and it is obvious that \eqref{eq:int_fact} holds for $m=1$. By
\eqref{eq:int_fact}, we can rewrite \eqref{eq:DysonG} as
\begin{equation} \label{eq:Gsplit}
\begin{split}
G(\Sf, \Si) &= \sum_{\substack{m=0\\[2pt] m \text{ is even}}}^{+\infty}
  \sum_{p=0}^m \int_{\Sf > s_m > \cdots > s_{p+1} > \Sarr}
  \int_{\Sarr > s_p > \cdots > s_1 > \Si} \sum_{\mf{q} \in \mQ(\bs)} \\
& \hspace{60pt}(-1)^{\#\{\bs < t\}} \ii^m \mc{U}^{(0)}(\Sf, \bs, \Si)
  \mathcal{L}(\mf{q}) \,\dd s_1 \cdots \,\dd s_p
  \,\dd s_{p+1} \cdots \,\dd s_m \\
&= \sum_{\substack{m=0\\[2pt] m \text{ is even}}}^{+\infty}
    \int_{\Sarr > s_m > \cdots > s_1 > \Si} \sum_{\mf{q} \in \mQ(\bs)}
      (-1)^{\#\{\bs < t\}} \ii^m \mc{U}^{(0)}(\Sf, \bs, \Si)
      \mathcal{L}(\mf{q}) \,\dd s_1 \cdots \,\dd s_m \\
& \quad +\sum_{\substack{m=2\\[2pt] m \text{ is even}}}^{+\infty}
  \sum_{p=0}^{m-1} \int_{\Sf > s_m > \cdots > s_{p+1} > \Sarr}
  \int_{\Sarr > s_p > \cdots > s_1 > \Si} \sum_{\mf{q} \in \mQ(\bs)} \\
& \hspace{60pt}(-1)^{\#\{\bs < t\}} \ii^m \mc{U}^{(0)}(\Sf, \bs, \Si)
  \mathcal{L}(\mf{q}) \,\dd s_1 \cdots \,\dd s_p \,\dd s_{p+1} \cdots \,\dd s_m,
\end{split}
\end{equation}
where we have used the absolute convergence \eqref{eq:abs_convergence} to
ensure the validity of the second equality. When $s_m < \Sarr$, we have
$\mc{U}^{(0)}(\Sf, \bs, \Si) = G_s^{(0)}(\Sf, \Sarr) \mc{U}^{(0)}(\Sarr, \bs, \Si)$. Hence
\begin{displaymath}
\begin{split}
& \sum_{\substack{m=0\\[2pt] m \text{ is even}}}^{+\infty}
  \int_{\Sarr > s_m > \cdots > s_1 > \Si} \sum_{\mf{q} \in \mQ(\bs)}
    (-1)^{\#\{\bs < t\}} \ii^m \mc{U}^{(0)}(\Sf, \bs, \Si)
    \mathcal{L}(\mf{q}) \,\dd s_1 \cdots \,\dd s_m \\
={} & G_s^{(0)}(\Sf, \Sarr)
  \sum_{\substack{m=0\\[2pt] m \text{ is even}}}^{+\infty}
  \int_{\Sarr > s_m > \cdots > s_1 > \Si} \sum_{\mf{q} \in \mQ(\bs)}
    (-1)^{\#\{\bs < t\}} \ii^m \mc{U}^{(0)}(\Sarr, \bs, \Si)
    \mathcal{L}(\mf{q}) \,\dd s_1 \cdots \,\dd s_m
= G_s^{(0)}(\Sf,\Sarr) G(\Sarr, \Si).
\end{split}
\end{displaymath}
Inserting this equation into \eqref{eq:Gsplit} yields our conclusion
\eqref{eq:DysonArrow}.
\end{proof}

By the above lemma, we have extracted the first diagram in \eqref{eq:theorem} from the definition of the bold line \eqref{eq:DysonG_diagram}. Our next step is to introduce the inchworm proper pairings to the series expansion. The basic idea is to decompose every set of pairs into the union of one inchworm proper subset and several other unlinked subsets. The result reads
\begin{lemma} \label{lem:Dyson}
When the Dyson series \eqref{eq:DysonG} is absolutely convergent in the sense
of \eqref{eq:abs_convergence}, for any $\Sarr \in (\Si, \Sf)$, it holds that
\begin{equation} \label{eq:full_exp}
\begin{split}
G(\Sf,\Si) &= G_s^{(0)}(\Sf, \Sarr) G(\Sarr, \Si) +
  \sum_{\substack{\tilde{m}=2\\[2pt] \tilde{m} \text{ is even}}}^{+\infty}
  \sum_{\tilde{p}=0}^{\tilde{m}-1}
  \int_{\Sf > \tilde{s}_{\tilde{m}} > \cdots > \tilde{s}_{\tilde{p}+1} > \Sarr}
  \int_{\Sarr > \tilde{s}_{\tilde{p}} > \cdots > \tilde{s}_1 > \Si} \\
& \qquad \sum_{\substack{p=0\\[1pt] \tilde{p}-p \text{ is even}}}^{\tilde{p}}
  \sum_{\substack{n_p=0\\[1pt] n_p \text{ is even}}}^{\tilde{p}-p}
  \sum_{\substack{n_{p-1}=0\\[1pt] n_{p-1} \text{ is even}}}^{\tilde{p}-p-n_p}
  \cdots
  \sum_{\substack{n_1=0\\[1pt] n_1 \text{ is even}}}^{\tilde{p}-p-n_p-\cdots-n_2}
  \sum_{\mf{q} \in \mQ_{\Sarr}(\bs)} \sum_{\mf{q}_p \in \mQ(\bs^{(p)})}
    \cdots \sum_{\mf{q}_0 \in \mQ(\bs^{(0)})} \qquad \\
& \omit \hfill $\qquad (-1)^{\#\{\tilde{\bs} < t\}} \ii^{\tilde{m}}
  \mc{U}^{(0)}(\Sf, \tilde{\bs}, \Si) \mc{L}(\tilde{\mf{q}})
  \,\dd \tilde{s}_1 \cdots \,\dd \tilde{s}_{\tilde{p}}
  \,\dd \tilde{s}_{\tilde{p}+1} \cdots \,\dd \tilde{s}_{\tilde{m}}$,
\end{split}
\end{equation}
where the notations are
\begin{itemize}
\item $\tilde{\bs} = (\tilde{s}_{\tilde{m}}, \cdots, \tilde{s}_1)$;
\item $\tilde{\mf{q}} = \mf{q} \cup \mf{q}_p \cup \cdots \cup \mf{q}_0$;
\item $\bs, \bs^{(p)}, \cdots, \bs^{(0)}$ are all subsequences of
  $\tilde{\bs}$, defined by
  \begin{displaymath}
  \tilde{\bs} = (s_m, \cdots, s_{p+1}, \bs^{(p)}, s_p,
    \bs^{(p-1)}, s_{p-1}, \cdots, \bs^{(1)}, s_1, \bs^{(0)}),
  \end{displaymath}
  where $m = \tilde{m} + p - \tilde{p}$ and $\bs = (s_m, \cdots, s_1)$.
\end{itemize}
\end{lemma}

\begin{proof}
By comparing \eqref{eq:full_exp} with \eqref{eq:DysonArrow}, we construct the
following map for given $\tilde{m}$, $\tilde{p}$ and a decreasing sequence
$\tilde{\bs}$:
\begin{equation} \label{eq:map}
\begin{array}{clc}
\mc{P} & \rightarrow & \mQ(\tilde{\bs}) \\
(p,n_p,\cdots,n_1,\mf{q},\mf{q}_p,\cdots,\mf{q}_0) &
  \mapsto & \mf{q} \cup \mf{q}_p \cup \cdots \cup \mf{q}_0
\end{array}
\end{equation}
where
\begin{displaymath}
\begin{split}
\mc{P} = \{(p,n_p,\cdots,n_1,\mf{q},\mf{q}_p,\cdots,\mf{q}_0) \mid {}
  & p \in \{\tilde{p},\tilde{p}-2,\cdots,\tilde{p}-2\lfloor \tilde{p}/2 \rfloor\};
    \ n_1 + \cdots + n_p \leqslant \tilde{p}-p; \\
  & n_p,\cdots,n_1 \text{ are even}; \ \mf{q} \in \mQ_{\Sarr}(\bs);
    \ \mf{q}_p \in \mQ(\bs^{(p)}); \ \cdots; \ \mf{q}_0 \in \mQ(\bs^{(0)}) \}.
\end{split}
\end{displaymath}
The subsequences $\bs, \bs^{(p)}, \cdots, \bs^{(0)}$ are defined as in the
lemma. If \eqref{eq:map} is a bijection, then \eqref{eq:full_exp} is a
rearrangement of the series \eqref{eq:DysonArrow}. The equality
\eqref{eq:full_exp} then follows by the absolute convergence of the Dyson
series.

Now it remains only to show that \eqref{eq:map} is one-to-one. The definition
of $\mc{P}$ shows that in a map $(p,n_p,\cdots,n_1,\mf{q},\mf{q}_p,\cdots,
\mf{q}_0) \mapsto \tilde{\mf{q}}$, if $\mf{q}$ and $\tilde{\mf{q}}$ are given,
other parameters are automatically determined by
\begin{equation} \label{eq:pqn}
\begin{aligned}
& p = 2\abs{\mf{q}} - (\tilde{m} - \tilde{p}), \qquad
  \mf{q}_0 = \{(\tau_1, \tau_2) \in \tilde{\mf{q}}
  \mid \tau_1 < \tau_2 < s_1\}, \\
& \mf{q}_k = \{(\tau_1, \tau_2) \in \tilde{\mf{q}}
  \mid s_k < \tau_1 < \tau_2 < s_{k+1}\}, \qquad k = 1,\cdots,p, \\
& n_k = 2\abs{\mf{q}_k}, \qquad k = 1,\cdots,p.
\end{aligned}
\end{equation}
Here $2\abs{\mf{q}}$ ($2\abs{\mf{q}_k}$) is actually the number of time points
in $\mf{q}$ ($\mf{q}_k$). Meanwhile, the set of pairs $\mf{q}$ satisfies
\begin{enumerate}
\item $\mf{q}$ is not linked to any other pair in $\tilde{\mf{q}}$;
\item $\mf{q} \in \mQ_{\Sarr}(\bs)$ and therefore contains
  $\tilde{s}_{\tilde{p}+1}, \cdots, \tilde{s}_{\tilde{m}}$.
\end{enumerate}
Now we consider an arbitrary set of pairs $\tilde{\mf{q}} \in \mQ(\tilde{\bs})$
with its linked component decomposition being $\tilde{\mf{q}} =
\tilde{\mf{q}}_1 \cup \cdots \cup \tilde{\mf{q}}_n$. The only subset of
$\tilde{\mf{q}}$ satisfying the above two conditions is
\begin{displaymath}
\mf{q} = \bigcup\:\{\tilde{\mf{q}}_k \mid \tilde{\mf{q}}_k
  \text{ contains a pair } (\tau_1, \tau_2) \text{ with } \tau_2 >
  \tilde{s}_{\tilde{p}}\}.
\end{displaymath}
Taking such a $\mf{q}$ and applying \eqref{eq:pqn} to find other parameters, we
obtain an inverse image of $\tilde{\mf{q}}$. This inverse image is unique due
to the uniqueness of $\mf{q}$, which yields that \eqref{eq:map} is a bijection.
\end{proof}

\begin{lemma} \label{lem:inchworm}
When the Dyson series \eqref{eq:DysonG} is absolutely convergent in the sense
of \eqref{eq:abs_convergence}, for any $\Sarr \in (\Si, \Sf)$, it holds that
\begin{equation} \label{eq:iwG}
\begin{split}
G(\Sf,\Si) &= G^{(0)}(\Sf, \Sarr) G(\Sarr, \Si) +
  \sum_{\substack{m=2\\[2pt] m \text{ is even}}}^{+\infty}
  \sum_{p=0}^{m-1}
  \sum_{\substack{n_p=0\\[1pt] n_p \text{ is even}}}^{+\infty} \cdots
  \sum_{\substack{n_0=0\\[1pt] n_0 \text{ is even}}}^{+\infty}
  \int_{\Sf > \tilde{s}_{\tilde{m}} > \cdots > \tilde{s}_{\tilde{p}+1} > \Sarr}
  \int_{\Sarr > \tilde{s}_{\tilde{p}} > \cdots > \tilde{s}_1 > \Si} \\
& \qquad \sum_{\mf{q} \in \mQ_{\Sarr}(\bs)}
  \sum_{\mf{q}_p \in \mQ(\bs^{(p)})} \cdots \sum_{\mf{q}_0 \in \mQ(\bs^{(0)})}
  (-1)^{\#\{\tilde{\bs} < t\}} \ii^{\tilde{m}}
  \mc{U}^{(0)}(\Sf, \tilde{\bs}, \Si) \mathcal{L}(\tilde{\mf{q}})
  \,\dd \tilde{s}_1 \cdots \,\dd \tilde{s}_{\tilde{p}}
  \,\dd \tilde{s}_{\tilde{p}+1} \cdots \,\dd \tilde{s}_{\tilde{m}},
\end{split}
\end{equation}
where $\tilde{m} = m + n_0 + \cdots + n_p$, and other parameters are defined 
in the same way as in Lemma \ref{lem:Dyson}.
\end{lemma}

\begin{proof}
It is not difficult to see that the map
\begin{displaymath}
\begin{array}{ccc}
\mc{P} & \rightarrow & \wt{\mc{P}} \\
(m,p,n_p,\cdots,n_0) & \mapsto & (\tilde{m}, \tilde{p}, p, n_p, \cdots, n_1)
\end{array}
\end{displaymath}
with $\tilde{m} = m+n_p+\cdots+n_0$ and $\tilde{p} = p+n_p+\cdots+n_0$ is a
bijection. Here
\begin{equation}
\begin{split}
\mc{P} &= \Big\{(m,p,n_p,\cdots,n_0) \,\Big\vert\, m,n_p,\cdots,n_0
  \text{ are positive and even; } p \in \{0,\cdots,m-1\} \Big\}, \\
\wt{\mc{P}} &= \Big\{(\tilde{m},\tilde{p},p,n_p,\cdots,n_1)
  \,\Big\vert\, \tilde{m},n_p,\cdots,n_1 \text{ are positive and even;} \\
& \hspace{112pt} \tilde{p} \in \{0,\cdots,\tilde{m}-1\}; \,
  p \in \{0,2,\cdots,2\lfloor \tilde{p}/2 \rfloor\}; \,
  p + n_p + \cdots + n_1 \leqslant \tilde{p} \Big\}.
\end{split}
\end{equation}
Hence \eqref{eq:iwG} is a rearrangement of the series \eqref{eq:full_exp}, and
therefore \eqref{eq:iwG} follows by the absolute convergence of the Dyson
series.
\end{proof}

Now we are ready to carry out the proof of the inchworm series
\eqref{eq:inchworm}:
\begin{proof}[Proof of Theorem \ref{thm:inchworm}]
Following the notations in Lemma \ref{lem:Dyson} and \ref{lem:inchworm}, we
have the following identities:
\begin{gather*}
\mc{U}^{(0)}(\bs) = \mc{U}^{(0)}(\Sf, s_m, \cdots, s_{p+1}, \Sarr)
  \mc{U}^{(0)}(\Sarr, \bs^{(p)}, s_p) W_s
    \mc{U}^{(0)}(s_p, \bs^{(p-1)}, s_{p-1}) W_s
    \cdots W_s \mc{U}^{(0)}(s_1, \bs^{(0)}, \Si), \\
\#\{\tilde{\bs} < t\}
  = \#\{\bs < t\} + \#\{\bs^{(p)} < t\} + \cdots + \#\{\bs^{(0)} < t\},
\qquad \mc{L}(\tilde{\mf{q}})
  = \mc{L}(\mf{q}) \mc{L}(\mf{q}_p) \cdots \mc{L}(\mf{q}_0).
\end{gather*}
Substituting these equalities to \eqref{eq:iwG} yields
\begin{equation} \label{eq:InchwormG}
\begin{split}
&G(\Sf, \Si) = G_s^{(0)}(\Sf,\Sarr) G(\Sarr,\Si)
+ \sum_{\substack{m=2\\[2pt] m \text{ is even}}}^{+\infty} \sum_{p=0}^{m-1}
  \sum_{\substack{n_p=0\\[1pt] n_p \text{ is even}}}^{+\infty} \cdots
  \sum_{\substack{n_0=0\\[1pt] n_0 \text{ is even}}}^{+\infty} \\
& \qquad \int_{\Sarr}^{\Sf} \int_{\Sarr}^{s_m} \cdots \int_{\Sarr}^{s_{p+2}}
  \int_{\Si}^{\Sarr} \int_{\Si}^{s_p} \cdots \int_{\Si}^{s_2}
  \int_{\bs^{(p)} \subset (s_p, \Sarr)}
  \int_{\bs^{(p-1)} \subset (s_{p-1}, s_p)} \cdots
  \int_{\bs^{(1)} \subset (s_1,s_2)}
  \int_{\bs^{(0)} \subset (\Si, s_1)} \\
& \qquad \sum_{\mf{q} \in \mQ_{\Sarr}(\bs)}
  \sum_{\mf{q}_p \in \mQ(\bs^{(p)})} \cdots \sum_{\mf{q}_0 \in \mQ(\bs^{(0)})}
    (-1)^{\#\{\bs < t\} + \#\{\bs^{(p)} < t\}
    + \cdots + \#\{\bs^{(0)} < t\}} \ii^{m+n_p+\cdots+n_0} \times{}\\
& \qquad \quad \mc{U}^{(0)}(\Sf, s_m, \cdots, s_{p+1}, \Sarr)
  \mc{U}^{(0)}(\Sarr, \bs^{(p)}, s_p) W_s
    \mc{U}^{(0)}(s_p, \bs^{(p-1)}, s_{p-1}) W_s
    \cdots W_s \mc{U}^{(0)}(s_1, \bs^{(0)}, \Si) \times {} \\
& \qquad \qquad \mc{L}(\mf{q}) \mc{L}(\mf{q}_p) \cdots \mc{L}(\mf{q}_0)
  \,\dd \bs^{(0)} \,\dd\bs^{(1)}
  \cdots \,\dd\bs^{(p-1)} \,\dd\bs^{(p)}
  \,\dd s_1 \cdots \,\dd s_{p-1} \,\dd s_p 
  \,\dd s_{p+1} \cdots \,\dd s_{m-1} \,\dd s_m,
\end{split}
\end{equation}
where the integrals with respect to $\bs^{(k)}$ are interpreted by
\begin{equation}
\int_{\bs^{(k)} \subset (a,b)} \dd \bs^{(k)} =
  \int_{b > s_{n_k}^{(k)} > \cdots > s_1^{(k)} > a}
    \,\dd s_1^{(k)} \cdots \,\dd s_{n_k}^{(k)}.
\end{equation}
Hence the integrals in the second line of \eqref{eq:InchwormG} is actually the
same as the integrals in \eqref{eq:iwG}. Due to the absolute convergence, we
are allowed to safely interchange sums and integrals, which can give us the
following factors:
\begin{equation} \label{eq:G_k}
\begin{aligned}
& \sum_{n_p = 0}^{+\infty} \int_{\bs^{(p)} \subset (s_p, \Sarr)}
  \sum_{\mf{q}_p \in \mQ(\bs^{(p)})} (-1)^{\#\{\bs^{(p)} < t\}} \ii^{n_k}
  \mc{U}^{(0)}(\Sarr, \bs^{(p)}, s_p) \mc{L}(\mf{q}_p) \,\dd \bs^{(p)}, \\
& \sum_{n_k = 0}^{+\infty} \int_{\bs^{(k)} \subset (s_k, s_{k+1})}
  \sum_{\mf{q}_k \in \mQ(\bs^{(k)})} (-1)^{\#\{\bs^{(k)} < t\}} \ii^{n_k}
  \mc{U}^{(0)}(s_{k+1}, \bs^{(k)}, s_k) \mc{L}(\mf{q}_k) \,\dd \bs^{(k)},
\qquad k = 0,\cdots,p-1.
\end{aligned}
\end{equation}
This quantity can be replaced, respectively, by $G(s_p, \Sarr)$ and $G(s_{k+1}, s_k)$, $k=0,\cdots,p-1$ according to \eqref{eq:DysonG}. Such replacement turns \eqref{eq:InchwormG} into
\begin{equation} \label{eq:inchworm1}
\begin{split}
& G(\Sf, \Si) = G_s^{(0)}(\Sf,\Sarr) G(\Sarr,\Si)
+ \sum_{\substack{m=2\\[2pt] m \text{ is even}}}^{+\infty}
  \sum_{p=0}^{m-1} \int_{\Sf > s_m > \cdots > s_{p+1} > \Sarr}
  \int_{\Sarr > s_p > \cdots > s_1 > \Si} \sum_{\mf{q} \in \mQ_{\Sarr}(\bs)}
  (-1)^{\#\{\bs < t\}} \ii^m \times {} \\
& \quad \mc{U}^{(0)}(\Sf, s_m, \cdots, s_{p+1}, \Sarr)
  G(\Sarr, s_p) W_s G(s_p, s_{p-1}) W_s
    \cdots W_s G(s_2, s_1) W_s G(s_1, \Si) \mc{L}(\mf{q})
  \,\dd s_1 \cdots \,\dd s_p \,\dd s_{p+1} \cdots \,\dd s_m.
\end{split}
\end{equation}
Noting that
\begin{displaymath}
\begin{split}
& \mc{U}^{(0)}(\Sf, s_m, \cdots, s_{p+1}, \Sarr) G(\Sarr, s_p) W_s G(s_p, s_{p-1}) W_s
  \cdots W_s G(s_2, s_1) W_s G(s_1, \Si) \\
={} & \mc{G}_{\Sarr}(\Sf, s_m) W_s \mc{G}_{\Sarr}(s_m, s_{m-1}) W_s \cdots
  W_s \mc{G}_{\Sarr}(s_2, s_1) W_s \mc{G}_{\Sarr}(s_1, \Si),
\end{split}
\end{displaymath}
we see that the integrand in \eqref{eq:inchworm1} is actually the same as the
integrand in \eqref{eq:inchworm}. Furthermore, the sum over $p$ in
\eqref{eq:inchworm1} can be replaced by $\sum_{p=0}^m$, because when $p = m$,
the set $\mQ_{\Sarr}(\bs)$ is empty. By doing this, we can apply the identity
\eqref{eq:int_fact} and obtain the equality \eqref{eq:inchworm}, which
completes the proof of Theorem \ref{thm:inchworm}.
\end{proof}

\added{%
\begin{remark}
In Theorem \ref{thm:inchworm}, when defining the absolute convergence, the
absolute value is outside the sum of $\mf{q}$ since $\mQ(\bs)$ is a finite set.
As mentioned in Remark \ref{rem:boson}, when the inchworm algorithm is applied
to fermions, due to the switching sign in $\mc{L}(\mf{q})$, the value of the
sum over $\mf{q}$ might become smaller due to cancellation; this would yield possibly improved convergence
in bold diagrammatic Monte Carlo methods for fermionic systems.
\end{remark}
}

\section{Integro-differential equations} \label{sec:id}
To better understand the algorithm, we are going to derive the limiting equation
of the inchworm algorithm by considering the case in which $\Sf - \Sarr$ is
infinitesimal. Suppose $\Delta t = \Sf - \Sarr$ and both $\Sf$ and $\Sarr$ are on 
the forward branch of the Keldysh contour, i.e. $\Sarr < \Sf < t$. By 
\eqref{eq:G0} and \eqref{eq:mc_U}, we can find that
\begin{equation} \label{eq:dt_exp}
\begin{gathered}
G_s^{(0)}(\Sf, \Sarr) = \ee^{-\ii \Delta t H_s} =
  I - \ii \Delta t H_s + O(\Delta t^2), \\
\mc{U}^{(0)}(\Sf, s_m, \Sarr) = 
  G_s^{(0)}(\Sf, s_m) W_s G_s^{(0)}(s_m, \Sarr) =
  W_s + O(\Delta t).
\end{gathered}
\end{equation}
In the inchworm method \eqref{eq:inchworm1}, the domain of the integral
\begin{displaymath}
\int_{\Sf > s_m > \cdots > s_{p+1} > \Sarr}
  \text{(integrand)} \,\dd s_{p+1} \cdots \,\dd s_{m-1} \,\dd s_m
\end{displaymath}
has volume $\frac{1}{(m-p)!}\Delta t^{m-p}$. Therefore all the terms with $p <
m-1$ are higher order in $\Delta t$, and \eqref{eq:inchworm1} can be rewritten
as
\begin{equation} \label{eq:linear_G}
\begin{split}
G(\Sf, \Si) &= G_s^{(0)}(\Sf,\Sarr) G(\Sarr,\Si)
+ \sum_{\substack{m=2\\[2pt] m \text{ is even}}}^{+\infty}
  \ii^m \int_{\Sarr}^{\Sf}
  \int_{\Sarr > s_{m-1} > \cdots > s_1 > \Si} \sum_{\mf{q} \in \mQ_{\Sarr}(\bs)}
  (-1)^{\#\{\bs < t\}} \mc{L}(\mf{q}) \times {} \\
& \qquad \mc{U}^{(0)}(\Sf, s_m, \Sarr) \mc{U}(\Sarr, s_{m-1}, \cdots, s_1, \Si)
  \,\dd s_1 \cdots \,\dd s_{m-1} \,\dd s_m + O(\Delta t^2) \\
&= (I - \ii \Delta t H_s) G(\Sarr,\Si)
  + \sum_{\substack{m=2\\[2pt] m \text{ is even}}}^{+\infty}
    \ii^m \int_{\Sarr}^{\Sf} \int_{\Sarr > s_{m-1} > \cdots > s_1 > \Si}
    \sum_{\mf{q} \in \mQ_{\Sarr}(\bs)}
    (-1)^{\#\{\bs < t\}} \mc{L}(\mf{q}) \times {} \\
& \qquad W_s \mc{U}(\Sarr, s_{m-1}, \cdots, s_1, \Si)
  \,\dd s_1 \cdots \,\dd s_{m-1} \,\dd s_m + O(\Delta t^2),
\end{split}
\end{equation}
where we have used the short hand
\begin{displaymath}
\mc{U}(\Sarr, s_{m-1}, \cdots, s_1, \Si) = G(\Sarr, s_{m-1}) W_s
  G(s_{m-1}, s_{m-2}) W_s \cdots W_s G(s_2, s_1) W_s G(s_1, \Si).
\end{displaymath}
Our assumption that $\Sarr < \Sf < t$ shows that $\#\{\bs < t\} = m$. Also, by
\eqref{eq:mc_L} and \eqref{eq:linked_mQ}, one sees that
\begin{displaymath}
\sum_{\mf{q} \in \mQ_{\Sarr}(\bs)} \mc{L}(\mf{q}) =
  \sum_{\mf{q} \in \mQ_{\Sarr}(\Sarr, s_{m-1}, \cdots, s_1)} \mc{L}(\mf{q})
  + O(\Delta t) =
  \sum_{\mf{q} \in \mQ_c(\Sarr, s_{m-1}, \cdots, s_1)} \mc{L}(\mf{q})
  + O(\Delta t).
\end{displaymath}
Thus on the right-hand side of \eqref{eq:linear_G}, the first-order term of the
integrand is actually independent of $s_m$, which can then be integrated out.
We write the result by moving $G(\Sarr,\Si)$ to the left-hand side and divide
both sides by $\Delta t$:
\begin{equation}
\begin{split}
\frac{G(\Sf,\Si) - G(\Sarr,\Si)}{\Delta t} &= -\ii H_s G(\Sarr,\Si)
- \sum_{\substack{m=2\\[2pt] m \text{ is even}}}^{+\infty}
  \ii^m \int_{\Sarr > s_{m-1} > \cdots > s_1 > \Si} \\
& \sum_{\mf{q} \in \mQ_c(\Sarr, s_{m-1}, \cdots, s_1)}
  (-1)^{m-1} \mathcal{L}(\mf{q}) W_s \mc{U}(\Sarr, s_{m-1}, \cdots, s_1, \Si)
  \,\dd s_1 \cdots \,\dd s_{m-1} + O(\Delta t).
\end{split}
\end{equation}
Taking the limit as $\Delta t \rightarrow 0$ and renaming $m$ to $m+1$, we
obtain the integro-differential equations for the propagator
$G(\cdot,\cdot)$:
\begin{equation} \label{eq:forward}
\begin{split}
\frac{\partial G(\Sarr, \Si)}{\partial \Sarr} &= -\ii H_s G(\Sarr,\Si) \\
& \quad - \sum_{\substack{m=1\\[2pt] m \text{ is odd}}}^{+\infty}
  \ii^{m+1} \int_{\Sarr > s_m > \cdots > s_1 > \Si}
  \sum_{\mf{q} \in \mQ_c(\Sarr, \bs)}
  (-1)^{\#\{\bs < t\}} \mc{L}(\mf{q})
  W_s \mc{U}(\Sarr, \bs, \Si) \,\dd s_1 \cdots \,\dd s_m.
\end{split}
\end{equation}
Here we have again used the fact that all the components of $\bs$ are less than
$t$. The equation \eqref{eq:forward} gives an integro-differential equation for
the full propagator $G(\Sarr,\Si)$ when $\Sarr < t$. If $\Sarr > t$, we can use
the same method to derive a similar integro-differential equation:
\begin{equation} \label{eq:backward}
\begin{split}
\frac{\partial G(\Sarr, \Si)}{\partial \Sarr} &= \ii H_s G(\Sarr,\Si) \\
& + \sum_{\substack{m=1\\[2pt] m \text{ is odd}}}^{+\infty}
  \ii^{m+1} \int_{\Sarr > s_m > \cdots > s_1 > \Si}
  \sum_{\mf{q} \in \mQ_c(\Sarr, \bs)}
  (-1)^{\#\{\bs < t\}} \mc{L}(\mf{q}) W_s \mc{U}(\Sarr, \bs, \Si)
  \,\dd s_1 \cdots \,\dd s_m.
\end{split}
\end{equation}
When $\Sarr = t$, one can see from \eqref{eq:G} that $G(\cdot, \Si)$ is
discontinuous, and it satisfies
\begin{equation} \label{eq:discontinuity}
\lim_{\Sarr \rightarrow t^+} G(\Sarr, \Si) = 
  O_s \lim_{\Sarr \rightarrow t^-} G(\Sarr, \Si).
\end{equation}
Combing \eqref{eq:forward} and \eqref{eq:backward}, we get the following
theorem:
\begin{theorem} \label{thm:integro-differential}
When the Dyson series \eqref{eq:DysonG} is absolutely convergent in the sense
of \eqref{eq:abs_convergence}, the full propagator $G(\cdot,\cdot)$ satisfies
the integro-differential equation
\begin{equation} \label{eq:id_eq}
\begin{split}
& \sgn(\Sarr-t) \frac{\partial G(\Sarr, \Si)}{\partial \Sarr} =
  \ii H_s G(\Sarr,\Si) + \sum_{\substack{m=1\\[2pt] m \text{ is odd}}}^{+\infty}
  \ii^{m+1} \int_{\Sarr > s_m > \cdots > s_1 > \Si} \\
& \qquad \sum_{\mf{q} \in \mQ_c(\Sarr, \bs)}
  (-1)^{\#\{\bs < t\}} \mathcal{L}(\mf{q}) W_s \mc{U}(\Sarr, \bs, \Si)
  \,\dd s_1 \cdots \,\dd s_m, \qquad \forall \Si \in [0,2t]\backslash \{t\},
  \quad \Sarr \in [\Si,2t]\backslash\{t\},
\end{split}
\end{equation}
with the jump condition \eqref{eq:discontinuity} and the ``initial condition''
$G(\Sarr,\Sarr) = \Id$.
\end{theorem}
Although the integro-differential equation \eqref{eq:id_eq} has been derived
from the inchworm method, the infinitesimal terms $O(\Delta t)$ or $O(\Delta
t^2)$ have not been rigorously verified. Below we are going to provide a
rigorous proof of Theorem \ref{thm:integro-differential} starting from the 
definition of $G(\cdot,\cdot)$.

\begin{proof}[Proof of Theorem \ref{thm:integro-differential}]
We will start the proof by deriving the dynamics of the propagator with a more
straightforward method, and then show that the result is equivalent to
\eqref{eq:id_eq}. Again we consider the case $\Sarr < t$, and take derivative
of the definition of $G$ \eqref{eq:G} to get
\begin{equation} \label{eq:diff}
\begin{split}
\frac{\partial G(\Sarr, \Si)}{\partial \Sarr} &=
  \tr_b(\rho_b G_b^{(0)}(2t, \Sarr) (\ii H_b - \ii H)
    \ee^{-\ii (\Sarr - \Si) H} G_b^{(0)}(\Si, 0)) \\
&= -\ii \tr_b(\rho_b G_b^{(0)}(2t, \Sf) (H_s + W)
    \ee^{-\ii (\Sarr - \Si) H} G_b^{(0)}(\Si, 0)) \\
&= -\ii H_s G(\Sarr, \Si) -\ii W_s \tr_b(\rho_b G_b^{(0)}(2t, \Sarr) W_b
    \ee^{-\ii (\Sarr - \Si) H} G_b^{(0)}(\Si, 0)).
\end{split}
\end{equation}
The propagator $\ee^{-\ii (\Sarr - \Si) H}$ can be expanded into Dyson series
as \eqref{eq:Dyson}, which turns \eqref{eq:diff} into
\begin{equation}
\begin{split}
& \frac{\partial G(\Sarr, \Si)}{\partial \Sarr} =
  -\ii H_s G(\Sarr, \Si) - \sum_{m=0}^{+\infty} \ii^{m+1}
    \int_{\Sarr > s_m > \cdots > s_1 > \Si} \\
& \quad (-1)^m W_s G_s^{(0)}(\Sarr, s_m) W_s G_s^{(0)}(s_m, s_{m-1}) W_s
    \cdots W_s G_s^{(0)}(s_2, s_1) W_s G_s^{(0)}(s_1, \Si) \times {} \\
& \quad \tr_b(\rho_b G_b^{(0)}(2t, \Sarr) W_b G_b^{(0)}(\Sarr, s_m) W_b
    G_b^{(0)}(s_m, s_{m-1}) W_b \cdots W_b G_b^{(0)}(s_2, s_1)
    W_b G_b^{(0)}(s_1, \Si) G_b^{(0)}(\Si, 0)) \\
& \omit \hfill $\,\dd s_1 \cdots \,\dd s_m.$
\end{split}
\end{equation}
Since $G_b^{(0)}(s_1, \Si) G_b^{(0)}(\Si, 0) = G_b^{(0)}(s_1, 0)$, we
can use the definitions \eqref{eq:mc_L} and \eqref{eq:mc_U} to simplify the
above equation:
\begin{equation} \label{eq:id_eq_Dyson}
\begin{split}
&\frac{\partial G(\Sarr, \Si)}{\partial \Sarr} =
  -\ii H_s G(\Sarr, \Si) - \sum_{m=0}^{+\infty} \ii^{m+1}
    \int_{\Sarr > s_m > \cdots > s_1 > \Si}
      (-1)^m W_s \mc{U}^{(0)}(\Sarr, \bs, \Si)
      \mathcal{L}(\Sarr, \bs) \,\dd s_1 \cdots \,\dd s_m \\
&\quad = -\ii H_s G(\Sarr, \Si) -
  \sum_{\substack{m=1\\[2pt] m\text{ is odd}}}^{+\infty} \ii^{m+1}
    \int_{\Sarr > s_m > \cdots > s_1 > \Si} \sum_{\mf{q} \in \mQ_c(\bs)}
    (-1)^{\#\{\bs < t\}} W_s \mc{U}^{(0)}(\Sarr, \bs, \Si)
    \mathcal{L}(\mf{q}) \,\dd s_1 \cdots \,\dd s_m.
\end{split}
\end{equation}
To see that the above equation is identical to the ``inchworm equation''
\eqref{eq:forward}, we can mimic \eqref{eq:InchwormG} and replace the
propagators $G(\cdot,\cdot)$ in the integral of \eqref{eq:forward} by its
Dyson series expansion. Following the same way as in the previous section, we
obtain a result similar to \eqref{eq:full_exp}:
\begin{equation} \label{eq:dG_dSarr}
\begin{split}
\frac{\partial G(\Sarr, \Si)}{\partial \Sarr} &= -\ii H_s G(\Sarr,\Si)
- \sum_{\substack{m=1\\[2pt] m \text{ is odd}}}^{+\infty}
  \sum_{\substack{n_m=0\\[2pt] n_m \text{ is even}}}^{+\infty} \cdots
  \sum_{\substack{n_0=0\\[2pt] n_0 \text{ is even}}}^{+\infty}
  \ii^{\tilde{m}+1}
  \int_{\Sarr > \tilde{s}_{\tilde{m}} > \cdots > \tilde{s}_1 > \Si} \\
& \quad \sum_{\mf{q} \in \mQ_c(\Sarr, \bs)}
  \sum_{\mf{q}_m \in \mQ(\bs^{(m)})} \cdots \sum_{\mf{q}_0 \in \mQ(\bs^{(0)})}
  (-1)^{\#\{\tilde{\bs} < t\}} W_s
  \mc{U}^{(0)}(\Sarr, \tilde{\bs}, \Si) \mc{L}(\tilde{\mf{q}})
  \,\dd \tilde{s}_1 \cdots \,\dd \tilde{s}_{\tilde{m}},
\end{split}
\end{equation}
where
\begin{displaymath}
\tilde{\mf{q}} = \mf{q} \cup \mf{q}_m \cup \cdots \cup \mf{q}_0, \qquad
\tilde{m} = m + n_m + \cdots + n_0, \qquad
\tilde{\bs} = (\tilde{s}_{\tilde{m}}, \cdots, \tilde{s}_1),
\end{displaymath}
and $\bs, \bs^{(m)}, \cdots, \bs^{(0)}$ can be determined by
\begin{displaymath}
\tilde{\bs} =
  (\bs^{(m)}, s_m, \bs^{(m-1)}, s_{m-1}, \cdots, \bs^{(1)}, s_1, \bs^{(0)}),
\qquad \bs = (s_m, \cdots, s_1).
\end{displaymath}
Now we need to use the following equivalence of sums:
\begin{displaymath}
\sum_{\substack{m=1\\[2pt] m \text{ is odd}}}^{+\infty}
  \sum_{\substack{n_m=0\\[2pt] n_m \text{ is even}}}^{+\infty} \cdots
  \sum_{\substack{n_0=0\\[2pt] n_0 \text{ is even}}}^{+\infty} =
\sum_{\substack{\tilde{m}=1\\[2pt] \tilde{m} \text{ is odd}}}^{+\infty}
  \sum_{\substack{m=0\\[2pt] \tilde{m}-m \text{ is even}}}^{\tilde{m}}
  \sum_{\substack{n_m=0\\[2pt] n_m \text{ is even}}}^{\tilde{m}-m} \cdots
  \sum_{\substack{n_{m-1}=0\\[2pt] n_{m-1} \text{ is even}}}^{\tilde{m}-m-n_m} \cdots
  \sum_{\substack{n_1=0\\[2pt] n_1 \text{ is even}}}^{\tilde{m}-m-n_m-\cdots-n_2}.
\end{displaymath}
Due to the absolute convergence of the Dyson series, the equation
\eqref{eq:dG_dSarr} becomes
\begin{equation}
\begin{split}
\frac{\partial G(\Sarr, \Si)}{\partial \Sarr} &= -\ii H_s G(\Sarr,\Si)
- \sum_{\substack{\tilde{m}=1\\[2pt] \tilde{m} \text{ is odd}}}^{+\infty}
  \ii^{\tilde{m}+1}
  \int_{\Sarr > \tilde{s}_{\tilde{m}} > \cdots > \tilde{s}_1 > \Si} \\
& \quad \sum_{\substack{m=1\\[2pt] \tilde{m}-m \text{ is even}}}^{\tilde{m}}
  \sum_{\substack{n_m=0\\[2pt] n_m \text{ is even}}}^{\tilde{m}-m}
  \sum_{\substack{n_{m-1}=0\\[2pt] n_{m-1} \text{ is even}}}^{\tilde{m}-m-n_m} \cdots
  \sum_{\substack{n_1=0\\[2pt] n_1 \text{ is even}}}^{\tilde{m}-m-n_m-\cdots-n_2}
  \sum_{\mf{q} \in \mQ_c(\Sarr, \bs)}
  \sum_{\mf{q}_m \in \mQ(\bs^{(m)})} \cdots \sum_{\mf{q}_0 \in \mQ(\bs^{(0)})} \\
& \omit \hfill $(-1)^{\#\{\tilde{\bs} < t\}} W_s
  \mc{U}^{(0)}(\Sarr, \tilde{\bs}, \Si) \mc{L}(\tilde{\mf{q}})
  \,\dd \tilde{s}_1 \cdots \,\dd \tilde{s}_{\tilde{m}}.$
\end{split}
\end{equation}
The last step is to verify that the right-hand side of the above equation
equals the right-hand side of \eqref{eq:id_eq_Dyson}. By comparison, we just
need to show that for any given positive odd integer $\tilde{m}$ and a sequence
$\tilde{\bs} = (s_{\tilde{m}}, \cdots, s_1)$ satisfying $\Sarr > s_{\tilde{m}}
> \cdots > s_1 > \Si$, the following map is a bijection:
\begin{equation}
\begin{array}{clc}
\wt{\mc{P}} & \rightarrow & \mQ_{\Sarr}(\tilde{\bs}) \\
(m,n_m,\cdots,n_0,\mf{q},\mf{q}_p,\cdots,\mf{q}_0) &
  \mapsto & \mf{q} \cup \mf{q}_p \cup \cdots \cup \mf{q}_0
\end{array}
\end{equation}
where
\begin{displaymath}
\begin{split}
\wt{\mc{P}} = \{(m,n_m,\cdots,n_0,\mf{q},\mf{q}_p,\cdots,\mf{q}_0) \mid {}
  & p \in \{\tilde{m},\tilde{m}-2,\cdots,1\};
    \ n_0 + \cdots + n_m = \tilde{m}-m; \ n_m,\cdots,n_0 \text{ are even}; \\
  & \mf{q} \in \mQ_c(\Sarr,\bs);
    \ \mf{q}_m \in \mQ(\bs^{(p)}); \ \cdots; \ \mf{q}_0 \in \mQ(\bs^{(0)}) \}.
\end{split}
\end{displaymath}
This map can actually be considered as a special case of \eqref{eq:map} when
$\tilde{p} = \tilde{m}-1$, and hence is one-to-one. Till now, the inchworm
equation \eqref{eq:id_eq} is confirmed to be identical to \eqref{eq:diff} when
$\Sarr < t$, which describes the correct dynamics of the full propagator. The
case $\Sarr > t$ can be dealt with following the same procedure, which is
omitted for conciseness.
\end{proof}

Since the equation in the above theorem is firstly derived by setting
$\Sf-\Sarr$ to be infinitesimal in the inchworm algorithm, the algorithm can
actually be considered as an iterative scheme for the equation. From this point
of view, we can improve the numerical method by solving \eqref{eq:id_eq} with a
higher-order scheme. Before introducing the details of the method, we will
first present the spin-boson model, and show that it satisfies all the
conditions needed for the inchworm method.

\added{%
\begin{remark}
The sum of integrals on the right-hand side of Eq.~\eqref{eq:id_eq}, can be
rewritten as
\begin{displaymath}
\int_{\Si}^{\Sarr} \Sigma(\Sarr, s_1) G(s_1, \Si) \,\mathrm{d}s_1
\end{displaymath}
with a properly defined $\Sigma(\Sarr, s_1)$, which makes the equation
\eqref{eq:id_eq} similar to the Kadanoff-Baym equations
\cite{Kadanoff1962}. It would be interesting to establish deeper
connections between these models, which we will leave for future works.
\end{remark}
}

\section{Spin-boson model} \label{sec:spin-boson}
To demonstrate the algorithm in a specific model, we consider the spin-boson
model in which the system is a single spin and the bath is given by a large
number of harmonic oscillators. In detail, we have
\begin{displaymath}
\mc{H}_s = \lspan\{ \ket{1}, \ket{2} \}, \qquad
\mc{H}_b = \bigotimes_{l=1}^L \left( L^2(\mathbb{R}^3) \right),
\end{displaymath}
where $L$ is the number of harmonic oscillators. The corresponding Hamiltonians
are
\begin{displaymath}
H_s = \epsilon \hat{\sigma}_z + \Delta \hat{\sigma}_x, \qquad
H_b = \sum_{l=1}^L \frac{1}{2} (\hat{p}_l^2 + \omega_l^2 \hat{q}_l^2).
\end{displaymath}
The notations are described as follows:
\begin{itemize}
\item $\epsilon$: energy difference between two spin states.
\item $\Delta$: frequency of the spin flipping.
\item $\hat{\sigma}_x, \hat{\sigma}_z$: Pauli matrices satisfying
  $\hat{\sigma}_x \ket{1} = \ket{2}$, $\hat{\sigma}_x \ket{2} = \ket{1}$,
  $\hat{\sigma}_z \ket{1} = \ket{1}$, $\hat{\sigma}_z \ket{2} = -\ket{2}$.
\item $\omega_l$: frequency of the $l$th harmonic oscillator.
\item $\hat{q}_l$: position operator for the $l$th harmonic oscillator defined
  by $\psi(q_1, \cdots, q_L) \mapsto q_l \psi(q_1, \cdots, q_L)$.
\item $\hat{p}_l$: momentum operator for the $l$th harmonic oscillator defined
  by $\psi(q_1, \cdots, q_L) \mapsto -\ii \nabla_{q_l} \psi(q_1, \cdots, q_L)$.
\end{itemize}
The coupling between system and bath is assumed to be linear:
\begin{displaymath}
W = W_s \otimes W_b, \qquad
  W_s = \hat{\sigma}_z, \qquad W_b = \sum_{l=1}^L c_l \hat{q}_l,
\end{displaymath}
where $c_l$ is the coupling intensity between the $l$th harmonic oscillator and
the spin. Suppose the initial state of the bath is in the thermal equilibrium
with inverse temperature $\beta$, \textit{i.e.} $\rho_b = Z^{-1} \exp(-\beta
H_b)$, and $Z$ is chosen such that $\tr(\rho_b) = 1$. Thus the hypothesis (H1)
is fulfilled, and Wick's theorem \eqref{eq:Wick} holds for
\begin{equation} \label{eq:B}
B(\tau_1, \tau_2) = \sum_{l=1}^L \frac{c_l^2}{2\omega_l} \left[
  \coth \left( \frac{\beta \omega_l}{2} \right) \cos \omega_l (\tau_2 - \tau_1)
  - \ii \sin \omega_l(\tau_2 - \tau_1)
\right].
\end{equation}

In order to apply the inchworm algorithm to the spin-boson model, we need to
show the absolute convergence \eqref{eq:abs_convergence}. By the definition of
$\mc{U}$ \eqref{eq:mc_U}, we immediately have
\begin{displaymath}
\| \mc{U}(\Sf,\bs,\Si) \|_s \leqslant \|W_s\|_s^m \max\{\|O_s\|_s, 1\}
  = \|\hat{\sigma}_z\|_s^m \max\{\|O_s\|_s,1\} = \max\{\|O_s\|_s,1\}.
\end{displaymath}
And from \eqref{eq:B}, we see that
\begin{equation} \label{eq:B_bound}
|B(\tau_1, \tau_2)| \leqslant \sum_{l=1}^L \frac{c_l^2}{2\omega_l}
  \sqrt{ \coth^2 \left( \frac{\beta \omega_l}{2} \right)
    \cos^2 \omega_l (\tau_2 - \tau_1) + \sin^2 \omega_l (\tau_2 - \tau_1)}
\leqslant \sum_{l=1}^L \frac{c_l^2}{2\omega_l}
  \coth \left( \frac{\beta \omega_l}{2} \right).
\end{equation}
Let $C_b$ be the right-hand side of the above inequality. Then when $m$ is even,
\begin{displaymath}
|\mc{L}(\mf{q})| \leqslant C_b^{m/2},
  \qquad \forall \mf{q} \in \mQ(s_m, \cdots, s_1).
\end{displaymath}
Since the number of pair sets in $\mQ(s_m, \cdots, s_1)$ is $(m-1)!!$, we have
\begin{equation} \label{eq:abs_convergence_spin_Boson}
\begin{split}
& \sum_{\substack{m=0\\[2pt] m \text{ is even}}}^{+\infty}
   \int_{\Sf > s_m > \cdots > s_1 > \Si} \|\mc{U}^{(0)}(\Sf, \bs, \Si)\|_s
     \left| \sum_{\mf{q} \in \mQ(\bs)} \mc{L}(\mf{q}) \right|
   \,\dd s_1 \cdots \,\dd s_m \\
\leqslant {} & \sum_{\substack{m=0\\[2pt] m \text{ is even}}}^{+\infty}
   \int_{\Sf > s_m > \cdots > s_1 > \Si}
     \max\{\|O_s\|_s, 1\} \cdot (m-1)!! C_b^{m/2} \,\dd s_1 \cdots \,\dd s_m \\
={} & \max\{\|O_s\|_s, 1\} \sum_{\substack{m=0\\[2pt] m \text{ is even}}}^{+\infty}
  \frac{(\Sf-\Si)^m}{m!!} C_b^{m/2}
= \max\{\|O_s\|_s, 1\} \exp \left( \frac{C_b(\Sf-\Si)^2}{2} \right).
\end{split}
\end{equation}
Since $\mc{H}_s$ is a finite-dimensional space, the observable $O_s$ is always a bounded operator. Therefore the right-hand side of the above equation is finite, 
which shows the absolute convergence. In the spin-boson model, people are usually interested in the population of the spin on each of the two spin states, meaning that we can take $O_s = \hat{\sigma}_z$.

When the Dyson series expansion \eqref{eq:DysonG} is directly used in the Monte
Carlo simulation, the fast growth of the variance as $\Sf - \Si$ increases
causes great numerical difficulties. Similar to
\eqref{eq:abs_convergence_spin_Boson}, the expectation of $\|\mc{U}^{(0)}(\Sf,
\bs, \Si) \mc{L}(\mf{q})\|_s^2$ can be estimated by
\begin{equation}
\begin{split}
& \sum_{\substack{m=0\\[2pt] m \text{ is even}}}^{+\infty}
  \int_{\Sf > s_m > \cdots > s_1 > \Si} \sum_{\mf{q} \in \mQ(\bs)}
    \|\mc{U}^{(0)}(\Sf, \bs, \Si)\|_s^2 |\mc{L}(\mf{q})|^2
  \,\dd s_1 \cdots \,\dd s_m \\
\leqslant {} & \sum_{\substack{m=0\\[2pt] m \text{ is even}}}^{+\infty}
  \int_{\Sf > s_m > \cdots > s_1 > \Si}
    (m-1)!! (\max\{\|O_s\|_s, 1\})^2 C_b^m \,\dd s_1 \cdots \,\dd s_m \\
={} & (\max\{\|O_s\|_s, 1\})^2 \exp \left( \frac{C_b^2 (\Sf-\Si)^2}{2} \right).
\end{split}
\end{equation}
Since $\|G(\Sf, \Si)\|_s \leqslant \max\{\|O_s\|_s, 1\}$, the growth of the variance is characterized by
\begin{displaymath}
\exp \left( \frac{C_b^2 (\Sf-\Si)^2}{2} \right) - 1.
\end{displaymath}
This is the well-known ``dynamical sign problem'' in the quantum Monte Carlo simulations. The inchworm method relieves the dynamical sign problem by lumping a number of samples based on the simulation results of shorter bold lines, which pushes the simulation time significantly longer.

Different from the inchworm methods presented in \cite{Chen2017a}, which directly applies the Monte Carlo method to \eqref{eq:inchworm}, we will design our numerical method based on the integro-differential equation \eqref{eq:id_eq}, and apply the idea of classical Runge-Kutta type methods for temporal discretization to solve the full propagators. As will be presented, this can both enhance the numerical efficiency and simplify the implementation. The algorithm will be detailed in the following section.

\section{Numerical method} \label{sec:num}
In order to find the full propagator $G(\Sf,\Si)$ numerically for all $\Si \in
[0,2t]$ and $\Sf \in [\Si, 2t]$, below we are going to develop a second-order
method in analogous to Heun's method for general ordinary differential
equations. For the general initial value problem
\begin{displaymath}
\frac{\mathrm{d}}{\mathrm{d}t} u(t) = f(t, u(t)), \qquad t > 0
\end{displaymath}
with initial condition $u(0) = u_0$, Heun's method reads
\begin{displaymath}
\begin{aligned}
U_k^* &= U_{k-1} + \Delta t f(t_{k-1}, U_{k-1}), \\
U_k &= \frac{1}{2} (U_{k-1} + U_k^*) + \frac{1}{2} \Delta t f(t_k, U_k^*),
\end{aligned}
\end{displaymath}
where $\Delta t$ is the time step, $t_k = t_{k-1} + \Delta t$, and $U_k$ is
the numerical approximation of $u(t_k)$. The method is second-order if the
solution is third-order continuously differentiable. In our case, the full
propagator $G(\cdot,\cdot)$ is known to be discontinuous on the line segments
$[0,t] \times \{t\}$ and $\{t\} \times [0,t]$ due to the presence of the
observable $O_s$.  Therefore in order to keep the second-order convergence
rate, special care needs to be taken for these discontinuities.

In our implementation, we take a uniform time step $\Delta t = t / N$, and
compute the numerical solutions $G(\Sf, \Si)$ only when $\Sf$ and $\Si$ are
multiples of $\Delta t$. This corresponds to a two-dimensional triangular mesh
shown in Figure \ref{fig:mesh} (for $N = 5$). Let $t_k = k \Delta t$ and
$G^{\Delta t}_{jk}$ be the numerical approximation of $G(t_j, t_k)$. Due to the
discontinuity, when $j = N$ (green line in Figure \ref{fig:mesh}) or $k = N$
(blue line in Figure \ref{fig:mesh}), $G^{\Delta t}_{jk}$ is considered to be
multiple-valued, and we will use the notation $j, k = N^+, N^-$ to define the
left and right limits. Precisely,
\begin{itemize}
\item For $k < N$, $G^{\Delta t}_{N^-k}$ and $G^{\Delta t}_{N^+k}$ are
  respectively the approximation of $\lim\limits_{s \rightarrow t^-} G(s,
  k\Delta t)$ and $\lim\limits_{s \rightarrow t^+} G(s, k\Delta t)$.
\item For $j > N$, $G^{\Delta t}_{jN^-}$ and $G^{\Delta t}_{jN^+}$ are
  respectively the approximation of $\lim\limits_{s \rightarrow t^-} G(j \Delta
  t, s)$ and $\lim\limits_{s \rightarrow t^+} G(j \Delta t, s)$.
\item $G^{\Delta t}_{N^- N^-} = G^{\Delta t}_{N^+ N^+} = \Id$ and $G^{\Delta
  t}_{N^+ N^-} = O_s$.
\end{itemize}

\begin{figure}[!ht]
\centering
\includegraphics[width=.4\textwidth]{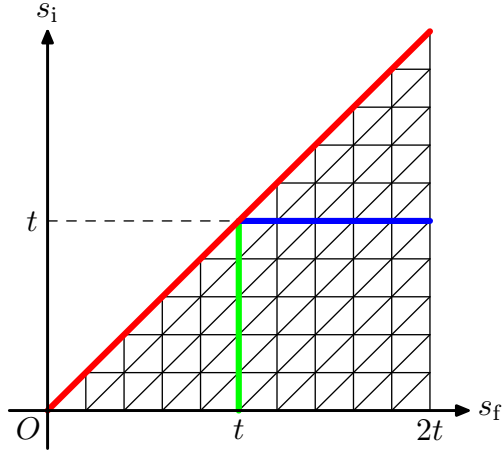}
\caption{The two-dimensional mesh to discretize $G(\Sf, \Si)$ for $N=5$}
\label{fig:mesh}
\end{figure}

Now we are ready to sketch our numerical algorithm. In general, the values of 
$G^{\Delta t}_{jk}$ are obtained in the following order:
\begin{equation} \label{eq:order}
\begin{aligned}
& G^{\Delta t}_{00}, \\
& G^{\Delta t}_{11}, \quad G^{\Delta t}_{10}, \\
& \cdots \quad \cdots \quad \cdots \quad \cdots \\
& G^{\Delta t}_{N-1,N-1}, \quad \cdots, \quad G^{\Delta t}_{N-1,0}, \\
& G^{\Delta t}_{N^-,N^-}, \quad G^{\Delta t}_{N^-,N-1},
  \quad \cdots, \quad G^{\Delta t}_{N^-,0}, \\
& G^{\Delta t}_{N^+,N^+}, \quad G^{\Delta t}_{N^+,N^-},
  \quad G^{\Delta t}_{N^+, N-1}, \quad \cdots, \quad G^{\Delta t}_{N^+,0}, \\
& G^{\Delta t}_{N+1,N+1}, \quad G^{\Delta t}_{N+1,N^+},
  \quad G^{\Delta t}_{N+1,N^-}, \quad G^{\Delta t}_{N+1, N-1},
  \quad \cdots, \quad G^{\Delta t}_{N+1,0}, \\
& \cdots \quad \cdots \quad \cdots \quad \cdots \quad
  \cdots \quad \cdots \quad \cdots \quad \cdots \quad
  \cdots \quad \cdots \quad \cdots \quad \cdots \quad \cdots \quad \cdots \\
& G^{\Delta t}_{2N,2N}, \quad \cdots, \quad G^{\Delta t}_{2N, N+1}, \quad
  G^{\Delta t}_{2N,N^+}, \quad G^{\Delta t}_{2N, N^-}, \quad
  G^{\Delta t}_{2N,N-1}, \quad \cdots, \quad G^{\Delta t}_{2N,0}.
\end{aligned}
\end{equation}
This corresponds to computing the values of $G(\cdot,\cdot)$ column by column
in Figure \ref{fig:mesh} with the green line and the blue line split to two
lines due to the discontinuity. We first list out the three special cases:
\begin{itemize}
\item If $j = k$ ($N^+$ is considered to be not equal to $N^-$), we set
  $G^{\Delta t}_{jk}$ to be $\Id$. This corresponds to the nodes on the red line
  in Figure \ref{fig:mesh}.
\item If $j = N^+$ and $k \neq j$, we set $G^{\Delta t}_{jk}$ to be $O_s
  G^{\Delta t}_{N^- k}$. This is applied to the nodes on the green line
  in Figure \ref{fig:mesh}.
\item If $k = N^-$ and $j \neq k$, we set $G^{\Delta t}_{jk}$ to be $G^{\Delta
  t}_{j N^+} O_s$. This is applied to the nodes on the blue line in Figure
  \ref{fig:mesh}.
\end{itemize}
For all other cases, we follow Heun's method and find the values of $G^{\Delta
t}_{jk}$ as follows:
\begin{enumerate}
\item Let
  \begin{equation} \label{eq:stage1}
  \begin{split}
    G_{jk}^* = G^{\Delta t}_{j-1,k} & + \sgn(t_j - t) \Delta t \Bigg[
        \ii H_s G^{\Delta t}_{j-1,k} +
        \sum_{\substack{m=1\\[2pt] m \text{ is odd}}}^{+\infty} \ii^{m+1}
          \int_{t_{j-1} > s_m > \cdots > s_1 > t_k}
          \sum_{\mf{q} \in \mQ_c(t_{j-1}, \bs)}
          (-1)^{\#\{\bs < t\}} \\
    & \times \mc{L}(\mf{q}) W_s G_I(t_{j-1}, s_m) W_s G_I(s_m, s_{m-1}) W_s
      \cdots W_s G_I(s_2, s_1) W_s G_I(s_1, t_k) \,\dd s_1 \cdots \,\dd s_m \Bigg],
  \end{split}
  \end{equation}
  where $G_I(\cdot,\cdot)$ is the interpolated function satisfying
  \begin{equation} \label{eq:G_I}
  G_I(t_{j'}, t_{k'}) = G^{\Delta t}_{j'k'}, \qquad
    \text{for all integers $j', k'$ satisfying }
    k \leqslant k' \leqslant j' \leqslant j-1.
  \end{equation}
  In our implementation, piecewise linear interpolation is adopted, and the
  function $G_I(\cdot,\cdot)$ is linear on each triangle in Figure
  \ref{fig:mesh}.
\item Set $G_{jk}^{\Delta t}$ to be
  \begin{equation} \label{eq:stage2}
  \begin{split}
    & \frac{1}{2} G^{\Delta t}_{j-1,k} + \frac{1}{2} G_{jk}^* +
      \frac{1}{2} \sgn(t_j - t) \Delta t \Bigg[
        \ii H_s G_{jk}^* +
        \sum_{\substack{m=1\\[2pt] m \text{ is odd}}}^{+\infty} \ii^{m+1}
          \int_{t_j > s_m > \cdots > s_1 > t_k}
          \sum_{\mf{q} \in \mQ_c(t_j, \bs)} \\
    & \qquad (-1)^{\#\{\bs < t\}} \mc{L}(\mf{q})
      W_s G_I^*(t_j, s_m) W_s G_I^*(s_m, s_{m-1}) W_s \cdots W_s
      G_I^*(s_2, s_1) W_s G_I^*(s_1, t_k) \,\dd s_1 \cdots \,\dd s_m \Bigg],
  \end{split}
  \end{equation}
  where $G^*_I(\cdot,\cdot)$ is the interpolated function satisfying
  $G^*_I(t_j,t_k) = G^*_{jk}$ and
  \begin{equation} \label{eq:Gstar_I}
  G^*_I(t_{j'}, t_{k'}) = G^{\Delta t}_{j'k'}, \quad
    \text{for all integers $j', k'$ satisfying }
    k \leqslant k' \leqslant j' \leqslant j
    \text{ and } (j',k') \neq (j,k).
  \end{equation}
  Again, the same piecewise linear interpolation is adopted.
\end{enumerate}
\added{Note that Heun's method provides more convenient implementation than
other second-order Runge-Kutta methods (e.g. mid-point method), since the
result of the first stage ($G_{jk}^*$) is a first-order prediction of $G(t_j,
t_k)$, so that in the second stage \eqref{eq:stage2}, the evaluation of the
integral requires no additional interpolation or extrapolation (which is needed
if $G_{jk}^*$ falls between grid points)}. When applying \eqref{eq:stage1} and
\eqref{eq:stage2}, we follow the rules as below:
\begin{itemize}
\item When $j = N^-$, $j-1$ is interpreted as $N-1$; when $j = N+1$, $j-1$ is
  interpreted as $N^+$.
\item $\sgn(t_{N^-} - t) = -1$.
\item The interpolation of $G_I$ and $G_I^*$ should respect such
  discontinuities. Therefore when $k = N$ or $l = N$, the equation
  \eqref{eq:G_I} should be interpreted by
  \begin{gather*}
  \lim_{s\rightarrow t^{\pm}} G_I(t_k, s) = G_{kN^{\pm}}^{\Delta t}, \qquad
  \lim_{s\rightarrow t^{\pm}} G_I(s, t_l) = G_{N^{\pm}l}^{\Delta t}, \\
  \lim_{s\rightarrow t^+}
    \lim_{\tilde{s} \rightarrow t^+} G_I(\tilde{s}, s) = 
  \lim_{s\rightarrow t^-}
    \lim_{\tilde{s} \rightarrow t^-} G_I(s, \tilde{s}) = \Id, \qquad
  \lim_{s\rightarrow t^+}
    \lim_{\tilde{s} \rightarrow t^-} G_I(s, \tilde{s}) = O_s.
  \end{gather*}
  The equation \eqref{eq:Gstar_I} should be similarly interpreted. Especially,
  when $j = N^+$, the term $G_I(t_{N^+}, s_m)$ in the integral should be
  interpreted as $O_s G_I(t_{N^-}, s_m)$.
\end{itemize}
The order of computation \eqref{eq:order} ensures that all information needed
in the two stages has been obtained beforehand.

\added{
The above algorithm can be further improved by noticing that the equation
\eqref{eq:id_eq} has a clear linear part $\ii H_s G(\Sarr,\Si)$.  Therefore,
when $H_s$ is a small matrix, the exponential Runge-Kutta method can be
applied. In general, for the initial value problem
\begin{displaymath}
\frac{\mathrm{d}}{\mathrm{d}t} u(t) = \mc{L}u(t) + f(t,u(t)), \qquad t > 0
\end{displaymath}
with $\mc{L}$ being a linear operator, the Heun's method can be applied to
$\ee^{-t \mc{L}} u$ and we get the numerical scheme
\begin{align*}
U_k^* &= \ee^{\Delta t \mc{L}} U_{k-1} + \Delta t f(t_{k-1}, U_{k-1}), \\
U_k &= \ee^{\Delta t \mc{L}} U_{k-1} + \frac{1}{2} \Delta t
  \left[ f(t_{k-1}, U_{k-1}) + \ee^{-\Delta t \mc{L}} f(t_k, U_k^*) \right].
\end{align*}
In the current case, the operator $\mc{L}$ corresponds to $\pm \ii H_s$.
Therefore, when the above scheme is applied, the first term in the second
equation is $\ee^{\pm\ii \Delta t H_s} G_{j-1,k}^{\Delta t} = G_s^{(0)}(j \Delta
t, (j-1) \Delta t) G_{j-1,k}^{\Delta t}$, which is exactly the first term in
the inchworm expansion \eqref{eq:inchworm}. This is particularly useful if the linear part of the equation is stiff. While for the spin-boson system, which is not stiff, we will be content with the non-exponential integrator described above. }

Now we consider the numerical computation of the infinite sums in
\eqref{eq:stage1} and \eqref{eq:stage2}. The numerical results in
\cite{Chen2017b} show that in the original inchworm method \eqref{eq:inchworm},
we can truncate the series at $m = M$ for some positive even integer $M$, and
obtain results with sufficient quality. In our method, the integer $M$ needs to
be odd and we perform the similar truncation by replacing the infinite sums in
\eqref{eq:stage1} and \eqref{eq:stage2} with
\begin{equation} \label{eq:truncation}
\sum_{\substack{m=1\\[2pt] m \text{ is odd}}}^M \ii^{m+1}
  \int_{t_{\mathrm{f}} > s_m > \cdots > s_1 > t_k}
  \sum_{\mf{q} \in \mQ_c(t_{\mathrm{f}}, \bs)} S(\mf{q},\bs)
  \,\dd s_1 \cdots \,\dd s_m,
\end{equation}
where $S(\mf{q},\bs)$ is the summand in \eqref{eq:stage1} or \eqref{eq:stage2},
and $t_{\mathrm{f}}$ is the corresponding $t_j$ or $t_{j+1}$. For every $m$,
the high-dimensional integral over $\bs$ and the sum over $\mf{q}$ are
evaluated using the Monte Carlo method. The sampling of $\bs$ can be done by
sampling a uniform distribution in $[t_k, t_{\mathrm{f}}]^m$, and then sort the
time points. As for $\mf{q}$, we first list out all the elements in $\mQ_c
(t_{\mathrm{f}}, \bs)$, and then pick a random one for each sample. Obviously
the number of pair sets in $\mQ_c(t_{\mathrm{f}}, \bs)$ depends only on the
value of $m$, and it has been given in \cite{Stein1978a} that this number can
be evaluated by
\begin{displaymath}
N_1 = 1, \qquad N_m = \frac{m-1}{2}
  \sum_{\substack{j=1\\[1pt] j \text{ is odd}}}^{m-2} N_j N_{m-1-j}.
\end{displaymath}
In \cite{Stein1978b}, it is proven that $N_m$ grows asymptotically as $m!!$. In
our numerical experiments, we are only concerned about a small $m$, and
therefore such a strategy is feasible.

As the end of this section, we will discuss briefly the difference between this
numerical method and the original inchworm method based on the Monte Carlo
simulation of \eqref{eq:inchworm}. Both numerical methods are time-stepping
methods which can be regarded as a numerical approximation of
\begin{equation} \label{eq:time_stepping}
G(\Sf, \Si) = G(\Sarr, \Si) + \int_{\Sarr}^{\Sf} \mathit{RHS}(s) \,\dd s,
\end{equation}
where $\mathit{RHS}(s)$ is the right-hand side of \eqref{eq:forward} or
\eqref{eq:backward}, with $\Sarr$ replaced by $s$.  From the diagrammatic
equation \eqref{eq:theorem}, one can see that the original inchworm algorithm
in principle allows to put an arbitrary number of points between $\Sarr$ and
$\Sf$ to evaluate the integral in \eqref{eq:time_stepping}, and all these
points are stochastic, which implies that in the time-stepping process, the
integral between two time steps is approximated using a Monte Carlo simulation;
while in our method, this is replaced by a numerical integration of
second-order convergence, which can be expected to be more efficient. One
possible benefit of the Monte Carlo integration from $\Sarr$ to $\Sf$ is that
when a sufficient number of samples are used, this integral can be evaluated up
to arbitrary precision; while in the Runge-Kutta integration, an error of
$O((\Sf - \Sarr)^{\alpha})$, where $\alpha$ depends on the order of the method,
is always there. However, this does not mean that the numerical error will
vanish as the number of samples increases in the original inchworm algorithm,
since another part of numerical error comes from the approximation of
$\mathit{RHS}(s)$, where the interpolation on the lattice shown in Figure
\ref{fig:mesh} is inevitable in both methods. Such error will only vanish as
the grid size tends to zero, but will not vanish as the number of samples tends
to infinity. Usually, it is sufficient to match the order of accuracy for
evaluating the integral in \eqref{eq:time_stepping} and the interpolation.
\added{This can also be achieved by using another formulation of the original
inchworm method \eqref{eq:DysonArrow}, and applying Gaussian quadrature to the
integral from $\Sarr$ to $\Sf$, which will involve multi-dimensional Gaussian
quadrature. The implementation is easier and cheaper by
using the Runge-Kutta integration strategy.} Another benefit of this new method
is that we can draw the samples for $\mf{q}$ more easily. Consider a diagram
with four points. The set $\mQ_c(s_1, s_2, s_3, s_4)$ contains only one set of
pairs $\{(s_1, s_3), (s_2, s_4)\}$, while the inchworm proper set of pairs has
9 possibilities:
\begin{itemize}
\item $\mQ_{\Sarr}(s_1, s_2, s_3, s_4) = \{\{(s_1, s_3), (s_2, s_4)\}\}$ if $s_3 < \Sarr \leqslant s_4$;
\item $\mQ_{\Sarr}(s_1, s_2, s_3, s_4) = \{\{(s_1, s_3), (s_2, s_4)\}, \{(s_1, s_4), (s_2, s_3)\}\}$ if $s_2 < \Sarr \leqslant s_3$;
\item $\mQ_{\Sarr}(s_1, s_2, s_3, s_4) = \{\{(s_1, s_3), (s_2, s_4)\}, \{(s_1, s_4), (s_2, s_3)\}, \{(s_1, s_2), (s_3, s_4)\}\}$ if $s_1 < \Sarr \leqslant s_2$;
\item $\mQ_{\Sarr}(s_1, s_2, s_3, s_4) = \{\{(s_1, s_3), (s_2, s_4)\}, \{(s_1, s_4), (s_2, s_3)\}, \{(s_1, s_2), (s_3, s_4)\}\}$ if $\Sarr \leqslant s_1$.
\end{itemize}
For the more time points, finding out all the inchworm proper pair sets is even
harder than just finding out $\mQ_c(\bs)$. This \added{also makes the
implementation of Runge-Kutta method cheaper than Gaussian quadrature.}

\added{
\begin{remark}
As mentioned in Remark \ref{rem:boson}, 
we focus on the bosonic
setting in this work. For fermionic settings, the number of inchworm proper diagrams also
grows asymptotically as $O\left( \left( \frac{m+1}{2} \right)! \right)$, as
indicated in \cite{Boag2018}. On the other hand, in fermionic settings, by a more
sophisticated combinatoric technique called inclusion-exclusion principle
 \cite{Boag2018}, the number of diagrams to be summed can be reduced
intrinsically to $O(m^3 \beta^m)$ for a constant $\beta < 2$. This is to some extent in the similar spirit of the reduced cost of diagrammatic summations achieved by the determinant diagrammatic Monte
Carlo method \cite{Rossi2017} for fermionic systems.
\end{remark}
}

\section{Numerical experiments} \label{sec:examples}
In our numerical experiments, we consider the spin-Boson model with a bath with
Ohmic spectral density, for which the frequency $\omega_l$ are distributed in
$[0, \omega_{\max}]$ as introduced in \cite{Makri1999}:
\begin{displaymath}
\omega_l = \omega_c
  \ln \left( 1 - \frac{l}{L} [1 - \exp(\omega_{\max} / \omega_c)] \right),
  \qquad l = 1,\cdots,L,
\end{displaymath}
where $\omega_c$ is the primary frequency to be specified later. The coupling
intensity $c_l$ is
\begin{displaymath}
c_l = \omega_l \sqrt{\frac{\xi \omega_c}{L} [1 - \exp(\omega_{\max}/\omega_c)]},
  \qquad l = 1,\cdots,L,
\end{displaymath}
with $\xi$ being the Kondo parameter. To compare our results with reference
solutions, we adopt the parameters provided in \cite{Kelly2013} where $L =
200$ and $\beta = 5\Delta^{-1}$. Different settings of bias, coupling intensity
and nonadiabaticity will be considered in our experiments. In all the numerical
tests, the maximum frequency $\omega_{\max}$ is set to be $4\omega_c$, and the
time step is chosen to be $\Delta t = 0.1$ if not otherwise specified.
Numerical results obtained by the QuAPI method \cite{Makri1995a, Makri1995b}
will be used as reference solutions.

\subsection{Experiments with changing bias}
We first choose $\omega_c = 2.5\Delta$ and $\xi = 0.2$, for which the amplitude
of the bath correlation function $B(\cdot,\cdot)$ is plotted in Figure
\ref{fig:B}. In our implementation, we precompute the values of
$B(\cdot,\cdot)$ up to a very high precision on a very fine grid, and in the
inchworm algorithm, we retrieve its value by linear interpolation when necessary.
The numerical results for $\epsilon = 0, \Delta$ and $2\Delta$ are
given in Figure \ref{fig:bias}. It turns out that $M = 3$ in
\eqref{eq:truncation} can already provide satisfying numerical results up to
time $t = 5\Delta^{-1}$, and the general behavior of the observable has been
well captured by the results of $M = 1$.
\begin{figure}[!ht]
\centering
\includegraphics[width=.33\textwidth]{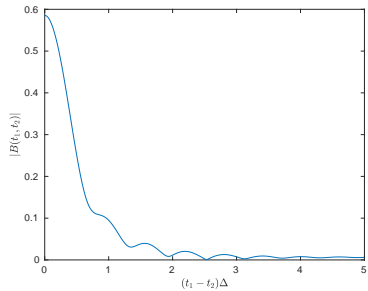}
\caption{The amplitude of the bath correlation function $B(\cdot, \cdot)$}
\label{fig:B}
\end{figure}

\begin{figure}[!ht]
\centering
\includegraphics[width=.33\textwidth]{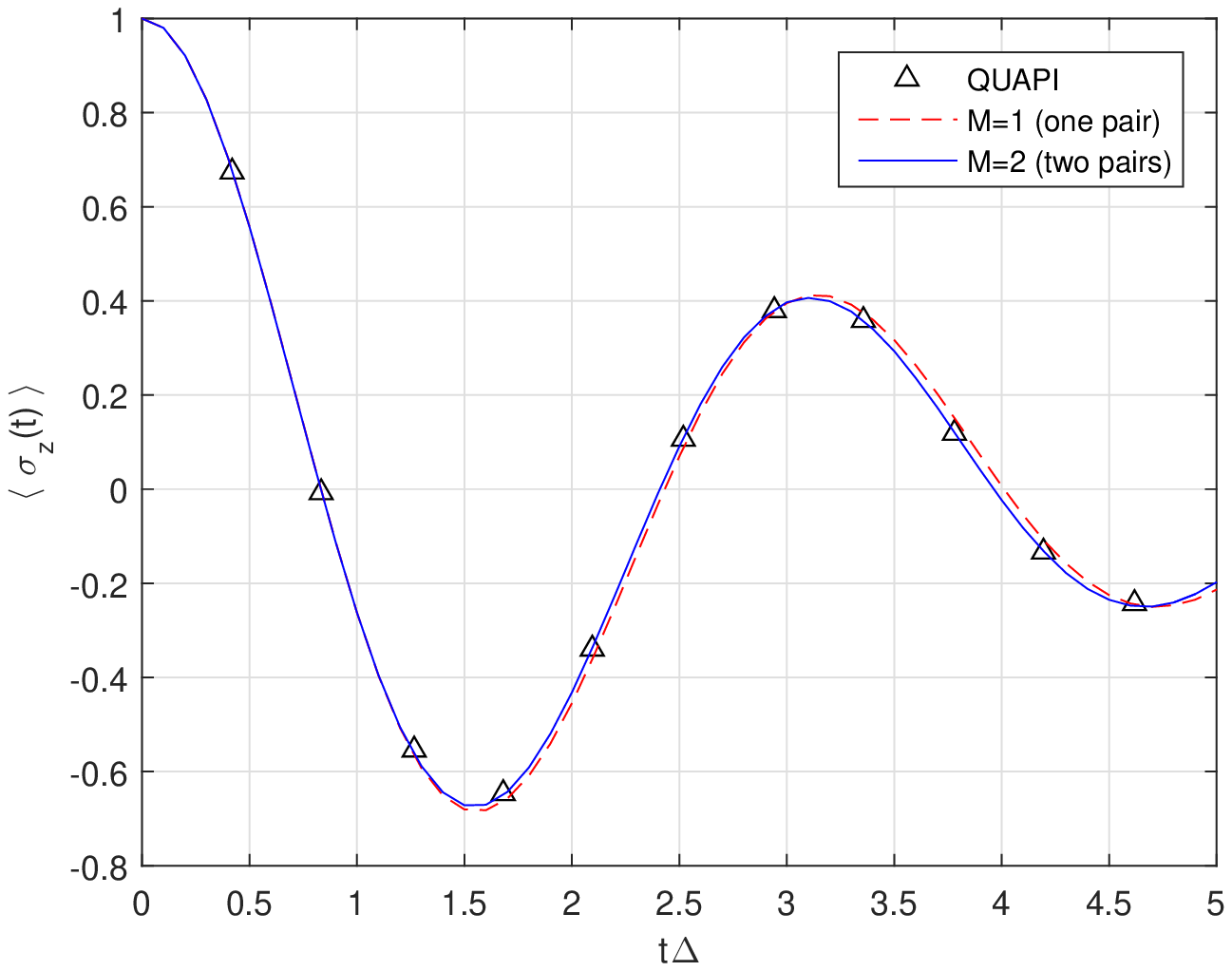}
\hspace{-10pt}
\includegraphics[width=.33\textwidth]{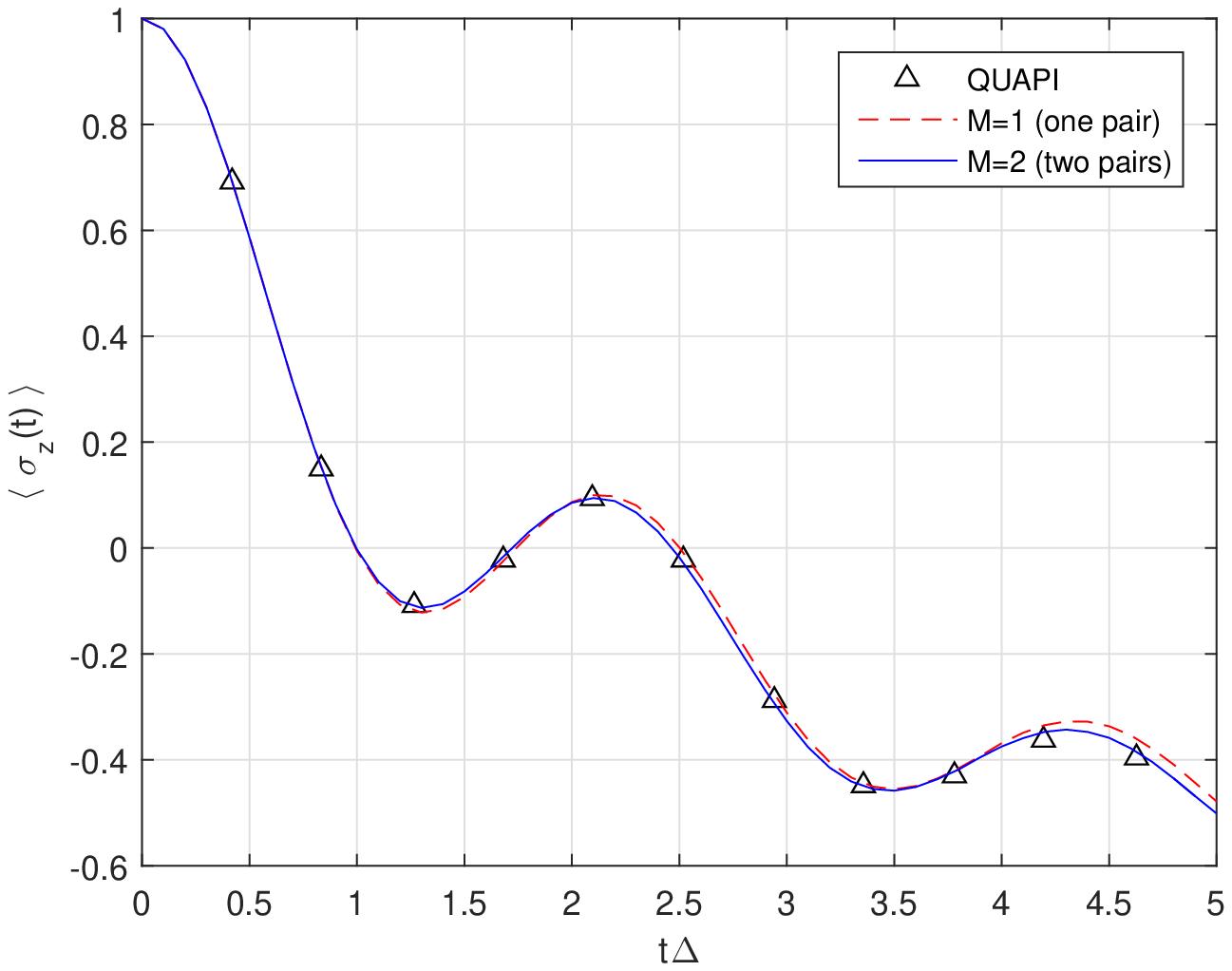}
\hspace{-10pt}
\includegraphics[width=.33\textwidth]{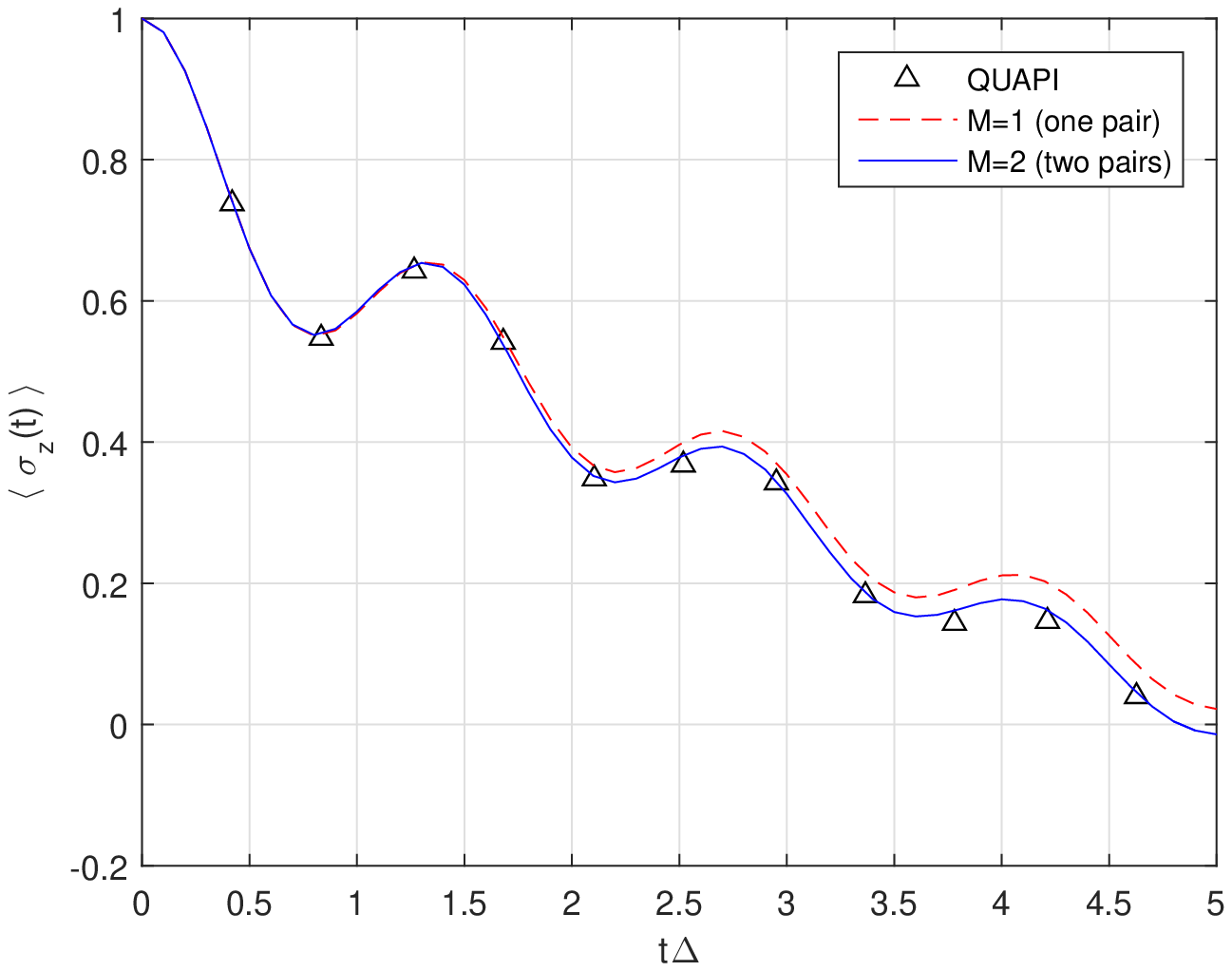}
\caption{Evolution of $\langle \sigma_z(t) \rangle$ under different settings of
the electronic bias (from left to right: $\epsilon = 0$, $\epsilon = \Delta$
and $\epsilon = 2\Delta$), with other parameters $\omega_c = 2.5\Delta$ and
$\xi = 0.2$. QuAPI results are plotted as a reference.}
\label{fig:bias}
\end{figure}

The order of convergence of our numerical method is also verified using the
test case with $\epsilon = 0$. In general, the stochastic error and the
``deterministic error'' caused by Runge-Kutta and interpolation cannot be
separated. In order to cast off the stochastic error, we only consider the
truncation $M = 1$, for which the sum \eqref{eq:truncation} contains only a
one-dimensional integral, and thus can be evaluated by the composite mid-point
rule. As a result, the whole scheme is deterministic and the order of
convergence is still expected to be $O(\Delta t^2)$. The reference solution is
obtained by choosing $\Delta t = 1/320 \Delta^{-1}$, and the numerical error of
$\langle \hat{\sigma}_z(t) \rangle$ is tabulated in Table \ref{tab:order},
which clearly shows the numerical error is second order of $\Delta t$.
\begin{table}[!ht]
\centering \small
\begin{tabular}{ c||c c|c c| c c|c c }
\hline
$h$ & error ($t=0.5$) & order & error ($t=1$) & order & error ($t=1.5$) & order&error ($t=2$) & order  \\
\hline
1/10 &  $0.0014$ &  -- & $0.0022$ & -- & $0.0014$ & -- & $0.0067$ & --\\
1/20 &  $0.0004$ &  $1.9740$ & $0.0005$ & $2.0537$ & $0.0004$ & $1.8855$ & $0.0017$ & $1.9910$\\
1/30 &  $0.1620\times 10^{-3}$ &  $1.9961$ & $0.2343\times 10^{-3}$ & $2.0307$ & $0.1722\times 10^{-3}$ & $1.9632$ & $0.7490\times 10^{-3}$ & $2.0064$ \\
1/40 &  $0.0909\times 10^{-3}$ &  $2.0087$ & $0.1305\times 10^{-3}$ & $2.0344$ & $0.0970\times 10^{-3}$ & $1.9941$ & $0.4189\times 10^{-3}$ & $2.0204$ \\
1/50 &  $0.0578\times 10^{-3}$ &  $2.0299$ & $0.0826\times 10^{-3}$ & $2.0470$ & $0.0618\times 10^{-3}$ & $2.0188$ & $0.2658\times 10^{-3}$ & $2.0376$ \\
1/60 &  $0.0398\times 10^{-3}$ &  $2.0505$ & $0.0567\times 10^{-3}$ & $2.0657$ & $0.0426\times 10^{-3}$ & $2.0422$ & $0.1826\times 10^{-3}$ & $2.0587$ \\
\hline
\end{tabular}
\caption{Numerical error of $\langle \hat{\sigma}_z(t) \rangle$ and the order of accuracy}
\label{tab:order}
\end{table}

\subsection{Experiments with changing coupling intensity}
Now we fix the values of $\omega_c$ and $\epsilon$ to be $2.5\Delta$ and
$\Delta$ respectively, and consider the coupling intensity $\xi = 0.1$ and $\xi
= 0.2$. It can be expected that the convergence of the inchworm series gets
slower when $\xi$ increases, as is confirmed in our numerical tests shown in
Figure \ref{fig:intensity}. In spite of this, very good matching with the QuAPI
results can still be obtained using $M = 1$ for both parameters, and further
improvement is indeed achieved by using $M = 3$.

As a reference, the numerical results for a direct summation of the truncated
Dyson series
\begin{displaymath}
G(\Sf, \Si) \approx \sum_{\substack{m=0\\[2pt] m \text{ is even}}}^M
   \int_{\Sf > s_m > \cdots > s_1 > \Si}
   \sum_{\mf{q} \in \mQ(\bs)} (-1)^{\#\{\bs < t\}} \ii^m \mc{U}^{(0)}(\Sf, \bs, \Si)
   \mc{L}(\mf{q}) \,\dd s_1 \cdots \,\dd s_m,
\end{displaymath}
are also provided in Figure \ref{fig:intensity} with label ``bare dQMC''. The
results show that the convergence of the Dyson series \eqref{eq:DysonG} is much
slower than the inchworm series, and therefore require a much larger $M$ to get
the same quality of the solutions for large $t$.

\begin{figure}[!ht]
\centering
\includegraphics[width=.33\textwidth]{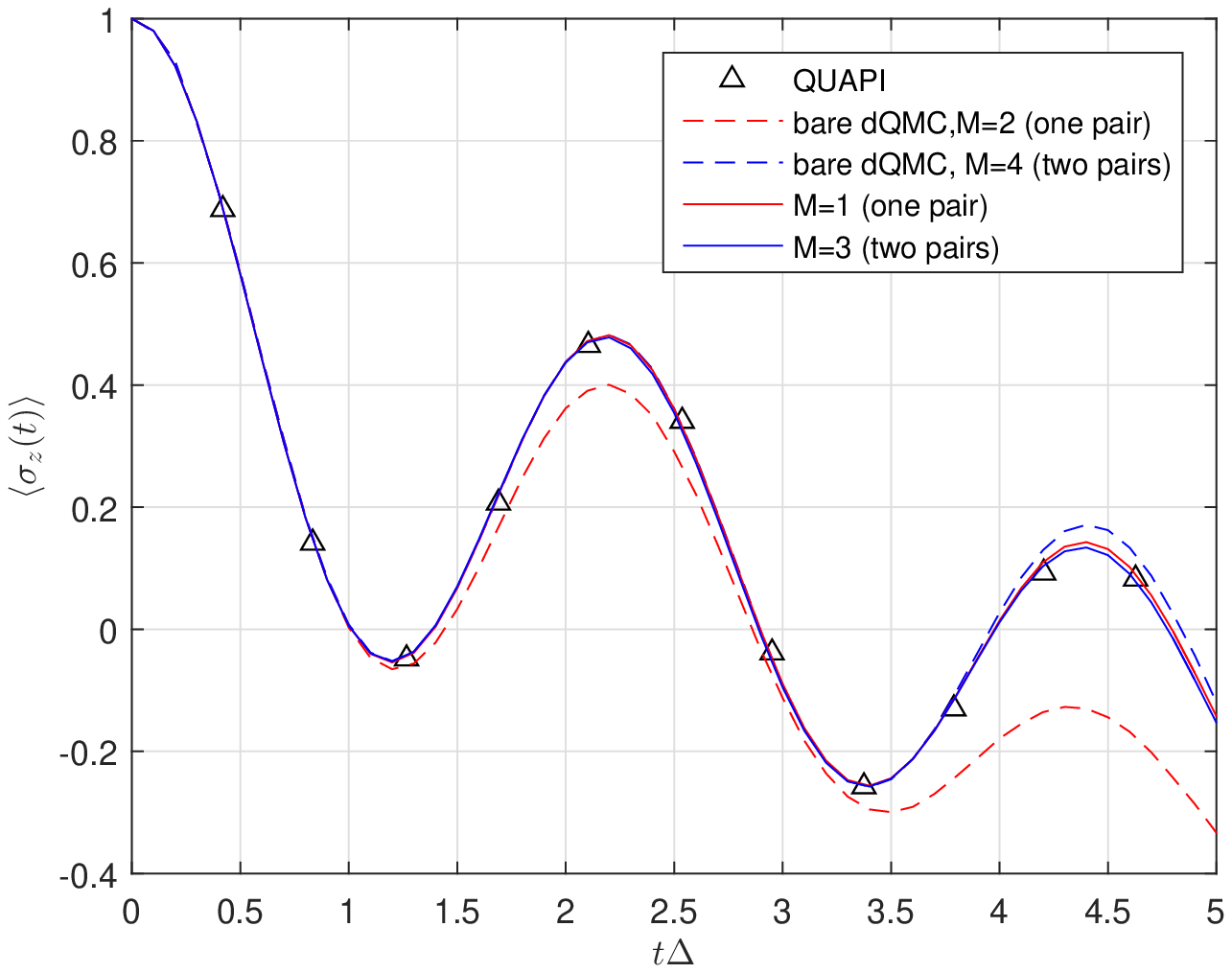}
\hspace{10pt}
\includegraphics[width=.33\textwidth]{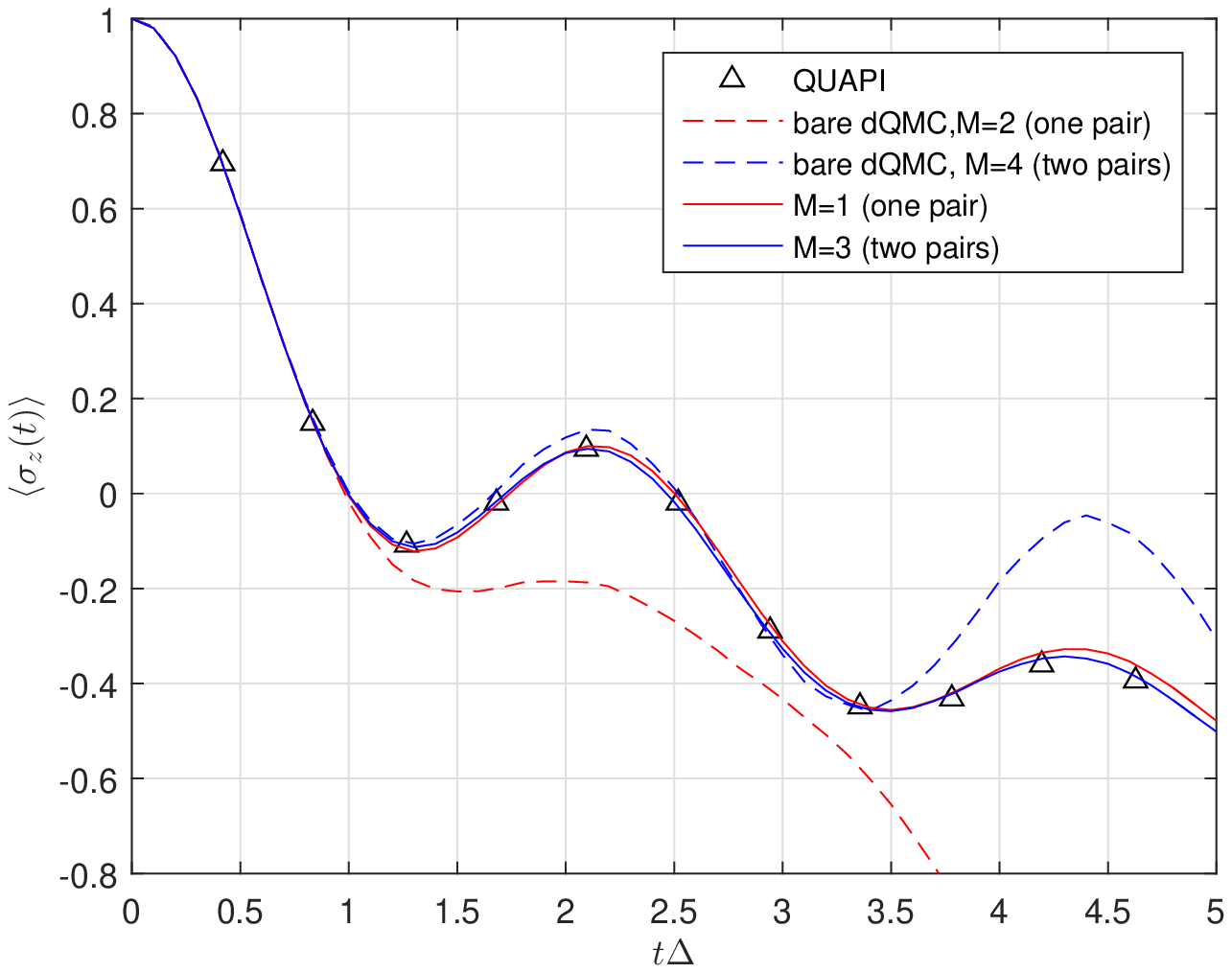}
\caption{Evolution of $\langle \sigma_z(t) \rangle$ under different settings of
the coupling intensity (left: $\xi = 0.1$, right: $\xi = 0.2$), with other
parameters $\omega_c = 2.5\Delta$ and $\epsilon = \Delta$. QuAPI results are
plotted as a reference.}
\label{fig:intensity}
\end{figure}

\subsection{Experiments with changing nonadiabaticity}
To show the role of the parameter $\omega_c$, we fix $\xi$ to be $0.4$ and
$\epsilon$ to be $\Delta$, and consider the cases $\omega_c = 0.25\Delta$ and
$\omega_c = 2.5\Delta$. Since the upper bound of $|B(\tau_1, \tau_2)|$ in
\eqref{eq:B_bound} gets larger when $\omega_c$ increases, we can expect slower
convergence in terms of $M$ for larger $\omega_c$. The evolution of the
observable is plotted in Figure \ref{fig:nonadiabaticity}. In the right figure,
due to the large value of both $\xi$ and $\omega_c$, even when $M = 3$ is used,
some discrepancy between our results and QuAPI can still be observed.

\begin{figure}[!ht]
\centering
\includegraphics[width=.33\textwidth]{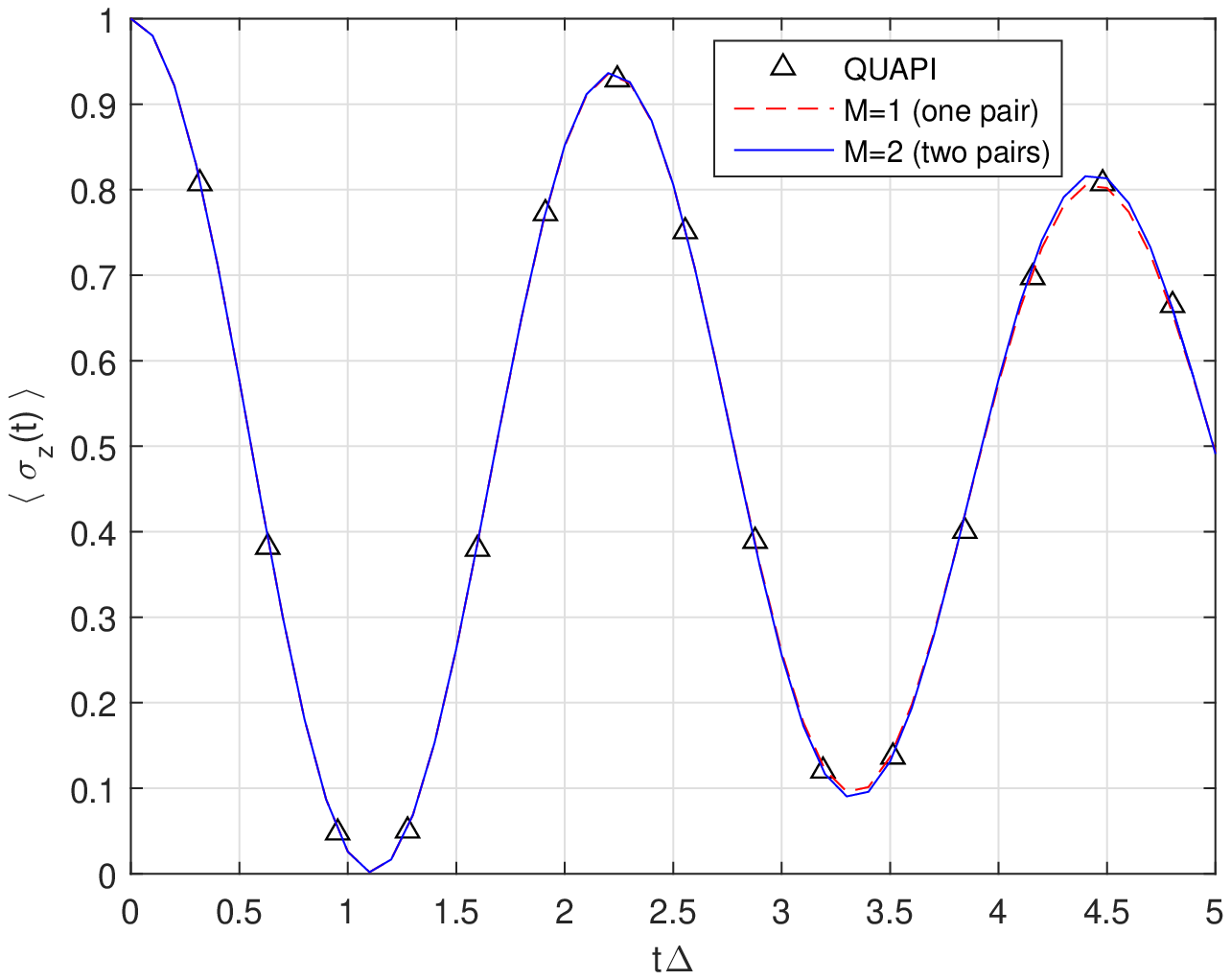}
\hspace{10pt}
\includegraphics[width=.33\textwidth]{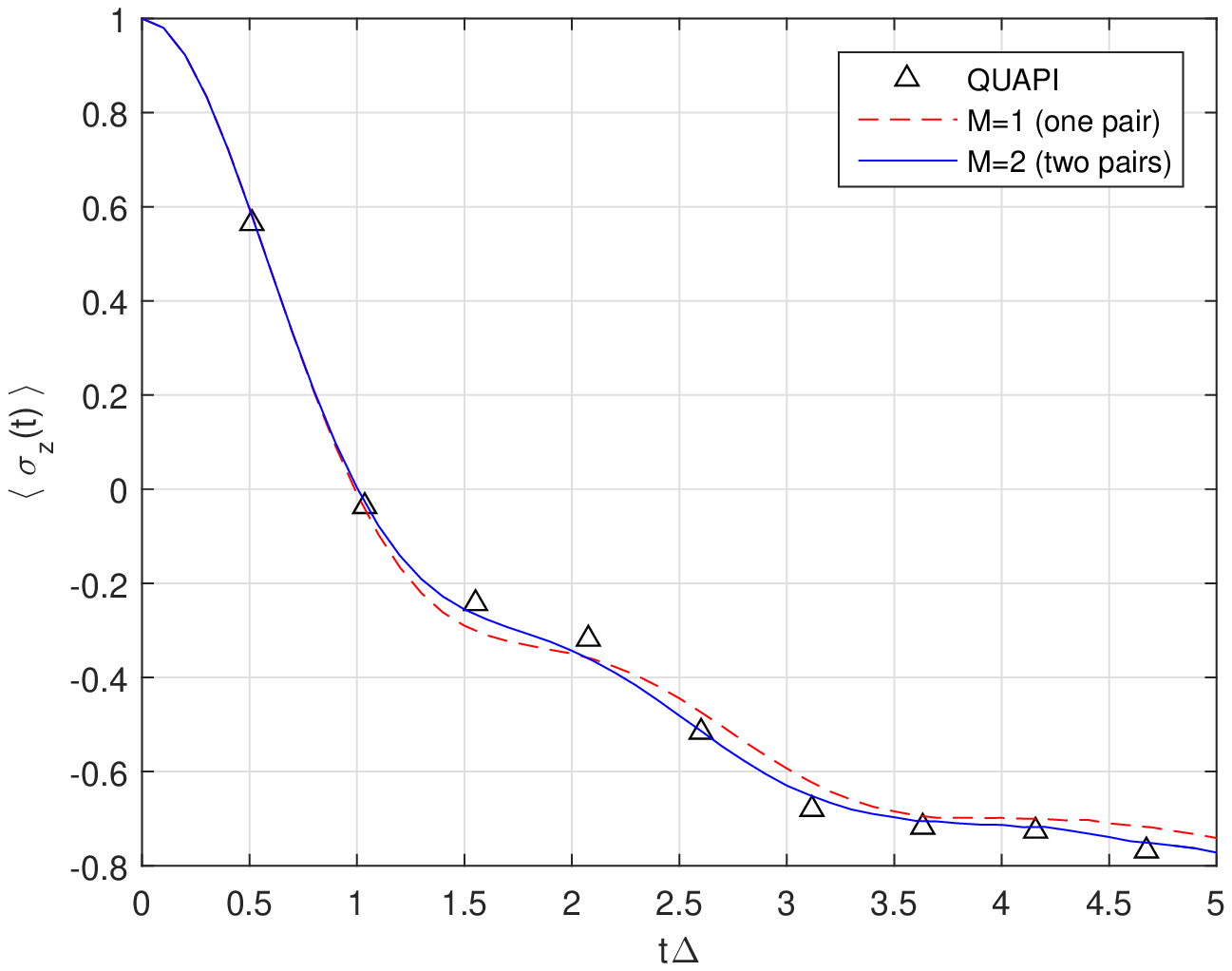}
\caption{Evolution of $\langle \sigma_z(t) \rangle$ under different settings of
the nonadiabaticity (left: $\omega_c = 0.25\Delta$, right: $\omega_c =
2.5\Delta$), with other parameters $\xi = 0.4$ and $\epsilon = \Delta$. QuAPI
results are plotted as a reference.}
\label{fig:nonadiabaticity}
\end{figure}

\section{Summary} \label{sec:summary} We have studied the inchworm
Monte Carlo method introduced in \cite{Cohen2015, Chen2017a}, and
proven rigorously that the method converges to the solution of the
open quantum system whose bath configuration satisfies Wick's
theorem. By assuming the iterative step length to be infinitesimal, we
have derived a new continuous model for the open quantum system, and
have proposed an improvement of the inchworm method to achieve better
numerical efficiency and simpler implementation. \added{Such a
framework is not restricted to the specific case considered in this
paper. In fact, the inchworm method has been applied to a number of
cases beyond the settings of this paper, including
\begin{itemize}
\item The fermionic bath \cite{Cohen2015, Antipov2017}, which needs a
  revision of the definition of $\mc{L}(\mf{q})$ \eqref{eq:Lq};
\item Full Keldysh contour \cite{Antipov2017}, which requires to
  augment Figure \ref{fig:Keldysh} with an imaginary time (equilibrium) part of the contour;
\item Two-time observables such as Green's functions
  \cite{Antipov2017}, which adds some additional restrictions on the
  admissible diagrams;
\item Non-unitary dynamics in the existence of an auxiliary counting
  field \cite{Ridley2018}, in which a non-Hermitian time-dependent
  Hermitian is introduced.
\end{itemize}
These extensions may complicate the mathematical formulation, while
the general ideas can be extended without much difficulty, and it can
be expected that new numerical methods based on integral-differential equation formalism can be derived accordingly.}
For the new numerical method, the stochastic error and the
deterministic error are coupled together, and one needs to apply
numerical analysis to find optimal combinations of parameters. \added{This will be an interesting future direction}.

\added{%
Possible applications of this method to dimension reduction of classical systems are also interesting. One interesting possibility is 
high-dimensional Langevin dynamics, especially when the system has a
conditional Gaussian structure, such as that considered in \cite{Chen2017}. Due to the similarity between
the Fokker-Planck operator and the Hamiltonian operator, the inchworm method
can be applied under certain conditions. More general problems may require
more sophisticated resummation techniques. These extensions are
left for future works.
}

\section*{Acknowledgements}
The work of Zhenning Cai is supported by National University of Singapore
Startup Fund under Grant No. R-146-000-241-133. The work of Jianfeng Lu is
supported in part by National Science Foundation under grant DMS-1454939. This
collaboration is also supported by National Science Foundation under grant
RNMS-1107444 (KI-Net).

\bibliographystyle{plain}
\bibliography{Inchworm}

\end{document}